\documentclass[twocolumn, a4paper,accepted=2022-11-28]{quantumarticle}

\pdfoutput=1

\usepackage[utf8]{inputenc}
\usepackage{xcolor}
\usepackage{tikz}
\usepackage{amsthm, amssymb}
\usepackage{amsmath}
\usepackage{braket}
\usepackage{graphicx}
\usepackage{enumerate}
\usepackage[frozencache,cachedir=.]{minted}
\usepackage[numbers,sort&compress]{natbib}
\usepackage[unicode, breaklinks]{hyperref}
\newtheorem{lemma}{Lemma}

\begin{document}

\title{Scaling of the quantum approximate optimization algorithm on superconducting qubit based hardware}
%\date{\today}

\author{Johannes Weidenfeller}
\affiliation{IBM Quantum, IBM Research Europe -- Zurich}
\affiliation{ETH Zurich} 
\orcid{0000-0003-4845-0476}
\author{Lucia C. Valor}
\affiliation{IBM Quantum, IBM Research Europe -- Zurich}
\orcid{0000-0003-1091-3172}
\author{Julien Gacon}
\affiliation{IBM Quantum, IBM Research Europe -- Zurich}
\affiliation{Institute of Physics, Ecole Polytechnique F\'ed\'erale de Lausanne (EPFL)}
\orcid{0000-0002-8617-1825}
\author{Caroline Tornow}
\affiliation{IBM Quantum, IBM Research Europe -- Zurich}
\affiliation{ETH Zurich}
\orcid{0000-0003-3214-9013}
\author{Luciano Bello}
\affiliation{IBM Quantum, IBM Research Europe -- Zurich}
\orcid{0000-0002-3530-6283}
\author{Stefan Woerner}
\affiliation{IBM Quantum, IBM Research Europe -- Zurich}
\orcid{0000-0002-5945-4707}
\author{Daniel J. Egger}
\email{deg@zurich.ibm.com}
\affiliation{IBM Quantum, IBM Research Europe -- Zurich}
\orcid{0000-0002-5523-9807}

\begin{abstract}
Quantum computers may provide good solutions to combinatorial optimization problems by leveraging the Quantum Approximate Optimization Algorithm (QAOA).
The QAOA is often presented as an algorithm for noisy hardware.
However, hardware constraints limit its applicability to problem instances that closely match the connectivity of the qubits.
Furthermore, the QAOA must outpace classical solvers.
Here, we investigate swap strategies to map dense problems into linear, grid and heavy-hex coupling maps.
A line-based swap strategy works best for linear and two-dimensional grid coupling maps.
Heavy-hex coupling maps require an adaptation of the line swap strategy.
By contrast, three-dimensional grid coupling maps benefit from a different swap strategy.
Using known entropic arguments we find that the required gate fidelity for dense problems lies deep below the fault-tolerant threshold.
We also provide a methodology to reason about the execution-time of QAOA.
Finally, we present a QAOA Qiskit Runtime program and execute the closed-loop optimization on cloud-based quantum computers with transpiler settings optimized for QAOA.
This work highlights some obstacles to improve to make QAOA competitive, such as gate fidelity, gate speed, and the large number of shots needed.
The Qiskit Runtime program gives us a tool to investigate such issues at scale on noisy superconducting qubit hardware.
\end{abstract}

\maketitle

\section{Introduction}

Gate-based quantum computers process information by applying unitary operations on information stored in qubits.
Such computers may provide an advantage for complex computational tasks in chemistry~\cite{Moll2018, Kandala2018, Ganzhorn2019}, finance~\cite{Woerner2019, Braine2019, Egger2020} and combinatorial optimization~\cite{Farhi2014, Farhi2014b}.
We focus on the Quantum Approximate Optimization Algorithm (QAOA)~\cite{Farhi2014, Farhi2014b, Yang2017} which maps combinatorial optimization problems, for instance,  quadratic unconstrained binary optimization (QUBO) problems with $n$ variables
\begin{align}
    \min_{x\in\{0,1\}^n}x^T\Sigma x, \: \Sigma \in \mathbb{R}^{n \times n},
\end{align}
to the problem of finding the ground state of an Ising Hamiltonian, $ H_C$~\cite{Farhi2014}.
Here, $ H_C$ is constructed by mapping each of the $n$ decision variables to a qubit by the relation $x_i=(1-z_i)/2$ and replacing $z_i$ by a Pauli spin operator $ Z_i$ to obtain $ H_C$~\cite{Lucas2014, Lodewijks2019}.
Two qubits $i,j$ thus only interact through $ Z_i Z_j$ if the corresponding quadratic term $\Sigma_{i,j}$ is not zero.
The QAOA first creates an initial state which is the ground state of a mixer Hamiltonian $ H_M$.
A common choice of $ H_M$ and initial state is $-\sum_{i=0}^{n-1} X_i$ and $\ket{+}^{\otimes n}$ which is easy to prepare.
Here, $X_i$ are Pauli $X$ operators.
Next, a depth-$p$ QAOA circuit creates the trial state $\ket{\psi(\boldsymbol{\beta}, \boldsymbol{\gamma})}$ for vectors $\boldsymbol{\gamma}, \boldsymbol{\beta}\in {\mathbb R}^p$ by applying $\exp{(-i\beta_k H_M)}\exp{(-i\gamma_k H_C)}$ at each layer $k=1,...,p$, implemented by $R_{X}(\beta)=\exp(-i\beta X/2)$ and $R_{ZZ}(\gamma)=\exp(-i\gamma ZZ/2)$ gates.
A classical optimizer seeks the optimal values of $\boldsymbol{\beta}$ and $\boldsymbol{\gamma}$ to create a trial state which minimizes the energy of $ H_C$.
The potential for a quantum advantage of QAOA and its variants over highly-optimized classical solvers, such as CPLEX~\cite{cplex}, must be explored empirically.
Such benchmarks must be two dimensional where both the quality and the time to reach the proposed solution matter~\cite{Kardashin2021, Harrigan2021, Sankar2021, Khumalo2021}.

Business-relevant problems often require budget or capacity constraints, and thus, $\Sigma$ tends to be dense~\cite{Markowitz1952} and the corresponding interaction graph  non-planar~\cite{Barahona1988}.
Implementing such problems on superconducting qubit~\cite{Devoret2013, Krantz2019} platforms is hindered by the limited qubit connectivity, expressed by the coupling map, and therefore requires SWAP gates~\cite{Harrigan2021}.
By contrast, cold-atomic architectures based on the Rydberg blockade~\cite{Isenhower2011} and trapped ions~\cite{Pagano2020} may overcome this issue~\cite{Zhou2020} but in turn suffer from low repetition rates which limits the speed at which shots are gathered~\cite{Labuhn2016}.

In this work we discuss all aspects relevant to scaling QAOA on superconducting qubits.
First, in Sec.~\ref{sec:connectivity}, we discuss strategies to map dense problems into linear, grid~\cite{Harrigan2021} and heavy-hex~\cite{Chamberland2020} coupling maps. 
Next, in Sec.~\ref{sec:hardware}, we estimate the quantum hardware requirements needed to solve such problems.
Here, we estimate in Sec~\ref{sec:fidelity} gate fidelity requirements for problems of varying density and present a methodology to reason about the run-time of QAOA in Sec.~\ref{sec:runtime}.
Indeed, QAOA can only provide an advantage by yielding better solutions than classical optimizers or comparable solutions in a  shorter time.
In Sec.~\ref{sec:runtime_program} we present a QAOA Qiskit Runtime program  to explore the QAOA scaling on noisy hardware.
Finally, we discuss these results and conclude in Sec.~\ref{sec:discussion}.

\begin{figure*}[htb!]
\includegraphics[width=\textwidth]{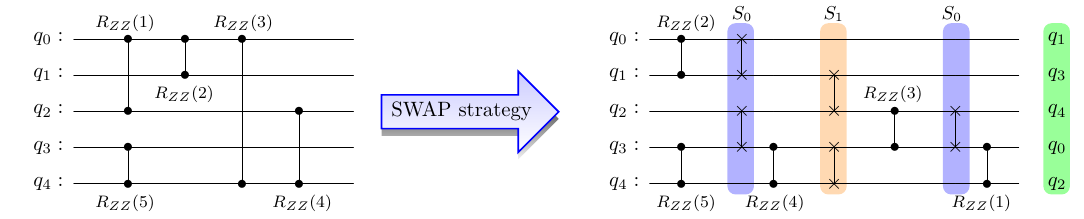}
\caption{
Transpilation of a five-qubit $\exp(-i\gamma H_C)$ circuit (left) to a line coupling map using $\mathcal{S}=\{S_0, S_1\}$. 
The swap layers alternate between the sets $S_0=\{{\rm SWAP}_{0,1}, {\rm SWAP}_{2,3}\}$ and $S_1=\{{\rm SWAP}_{1,2}, {\rm SWAP}_{3,4}\}$.
In the transpiled circuit (right) a redundant ${\rm SWAP}_{0,1}$ gate is removed from the last layer.
The resulting qubit mapping is highlighted in green.
}\label{fig:transpilation_example}
\end{figure*}

\section{Effects of limited device connectivity\label{sec:connectivity}}

The qubit connectivity in, e.g., superconducting qubit devices, is limited by engineering constraints and the need to avoid unwanted effects like cross-talk~\cite{Schutjens2013, McKay2019, Zhao2021}.
Typically, qubits are arranged in a planar graph, called the \emph{coupling map} and two-qubit gates can only be applied to adjacent qubits.
Therefore, additional SWAP gates are inserted into the circuits to make them hardware-compatible, a task known as qubit routing~\cite{Cowtan2019, Zulehner2019}.
The number of gates after transpilation to a hardware device thus depends on the problem, the coupling map, and the routing algorithm.
We propose a set of hardware dependent routing algorithms that perform particularly well on dense circuits of commuting operators, such as the cost operator in a QAOA on a complete graph.
We investigate the resulting circuit depth and gate count for linear, grid~\cite{Harrigan2021}, and heavy-hex~\cite{Chamberland2020} coupling maps.

\subsection{Hardware-optimized transpiler pass}\label{sec:hardware_optimized_transpiler_pass}
Swap transpiler passes typically divide quantum circuits into layers of simultaneously executable gates on the coupling map~\cite{Cowtan2019}.
They transition between the qubit mappings of different layers, i.e., a positioning of logical to physical qubits, by inserting SWAP gates consistent with the coupling map.
Here, a logical qubit is a qubit in an algorithm and a physical qubit is a hardware qubit such as a transmon~\cite{Koch2007}.
Mapping circuits to hardware is a hard optimization problem with combinatorial scaling, even on grid coupling maps~\cite{Lye2015}.
A variety of heuristic algorithms have therefore been introduced~\cite{Kole2016} and different coupling maps studied~\cite{Bhattacharjee2018, Farghadan2019}.
Application-specific transpilers leverage the mathematical properties of the application to reduce circuit depth.
For example, the 2QAN transpiler exploits the flexibility of permuting Trotter operators in two-local Hamiltonians~\cite{Lao2021}.
Swap synthesis is simpler when the considered gates commute, as in the cost layer of QAOA~\cite{Jin2021}.
Here, various strategies have been developed such as applying gates according to some ranking or stitching layers of $R_{ZZ}$ gates together~\cite{Alam2020}.
Transpiler passes that do not consider commutativity yield sub-optimal gate counts for circuits with a high number of commuting gates.

We develop a transpiler pass that first identifies the subset of all device qubits to run on and then applies a corresponding \emph{swap strategy}.
The swap strategy exploits gate commutativity by reordering commuting gates and inserting layers of SWAP gates from a set of predefined \emph{swap layers} $\mathcal{S}=\{S_0,...,S_K\}$.
Throughout this work a \emph{layer} is a set of simultaneously executable gates on the coupling map and has thus depth one.
Therefore, for a given coupling map, a swap strategy is a series of swap layers $S_{k_1}, S_{k_2}, ...$ of length $L_S$ applied in a predefined order and chosen from $\mathcal{S}$, i.e. $k_i\in\{0,...,K\}$.
A swap strategy applies the following steps:

\begin{enumerate} 
    \item Split the circuit into sequential sets of commuting gates $\mathcal{T}_1, \mathcal{T}_2, ...$, and choose the first one, i.e. $i=1$, as the current set.
    \item Repeat the following steps (a) to (d)  until all gates in the current set $\mathcal{T}_i$ are applied, see Fig.~\ref{fig:transpilation_example}. Set $j=1$, and
    \begin{enumerate}
        \item select all remaining gates $E_{j}\subseteq \mathcal{T}_i$ from the current set that are executable given the current qubit mapping and remove them from $\mathcal{T}_i$.
        \item Partition the selected gates $E_{j}$ into subsets of simultaneously executable gates $\mathcal{G}_1, \mathcal{G}_2, ...$ either by sorting them according to a provided edge coloring of the coupling map or by greedily building the sets.
        \item Iterate through the subsets $\mathcal{G}_1, \mathcal{G}_2, ...$, e.g., in decreasing set size, to simultaneously apply all gates in each set.
        \item Apply a single swap layer $\mathcal{S}_{f(i,j)}$ to alter the current qubit mapping. Here, $f(i,j)$ is the order in which we apply the swap layers. Increment $j$ and move to step (a) if $\mathcal{T}_i$ is not empty.
    \end{enumerate}
    \item Remove superfluous SWAP gates at the end of the circuit and continue to Step 2 with the next set of commuting gates $\mathcal{T}_{i+1}$, or terminate if all gates are applied.
\end{enumerate}
We call a swap strategy optimal if it leads to full connectivity with the least possible number of swap layers.
Our task is thus to find a good set of swap layers $\mathcal{S}$, the order in which to apply them $f(i,j)$, and the initial qubit mapping for a given problem and coupling map.
For example, in QAOA, the first set $\mathcal{T}_1$ creates the inital state which is trivial to apply as it is made of single-qubit gates.
Next, $\mathcal{T}_2$ corresponds to the cost operator which requires SWAP gates.
We note that for QAOA the order in which the gates are applied in Step 2(c) is chosen to leverage gate cancellations between $R_{ZZ}$ gates in $E_j$ and SWAP gates in $S_{f(i,j)}$.
In sub-sections~\ref{sec:swap_strategies} and \ref{sec:circuit_depth_and_gate_count} we present the swap strategies and the scaling of their gate count and depth with problem size and density, respectively.

\subsection{Swap strategies}\label{sec:swap_strategies}

The number of swap layers $L_S$ required to implement a quantum circuit $U$ of commuting two-qubit gates under a fixed swap strategy $\mathcal{S}$ depends on its density $D$ and structure.
Here, $D$ is the number of two-qubit terms normalized by its maximum possible number $n(n-1)/2$.
We describe $U$ with an interaction graph $G_\text{int}$ where vertices correspond to qubits and edges to two-qubit gates.
Let $G_{L_S}$ be the graph of all possible qubit interactions implementable after $L_S$ swap layers of $\mathcal{S}$.
The circuit can be implemented with $L_S$ swap layers of $\mathcal{S}$ if $G_\text{int}$ can be embedded in $G_{L_S}$.
The graphs $G_{L_S}$ therefore describe the structure of a potential $U$ that can be implemented after $L_S$ swap layers. 
The density of $G_{L_S}$ bounds the possible density of $U$ from above. 
Vice versa, we obtain a lower bound $L_S(D)$ on the number of swap layers required to implement a given $U$ with a density $D$ using a particular swap strategy. 
Indeed, a $G_{L_S}$ that achieves the same density as $U$ does not necessarily have the required structure to implement $U$.
We now discuss swap strategies that reach full connectivity, i.e. $D=1$.

\subsubsection{Linear coupling map\label{sec:line_strat}}
Transpiling arbitrary quantum circuits to a line coupling map has been studied in Ref.~\cite{Hirata2009} and when all the gates commute it can also be done optimally~\cite{Harrigan2021, Jin2021}.
For a line coupling map with $n$ qubits the swap strategy which alternates between two swap layers $S_0$ and $S_1$ which apply SWAP gates on all even and odd numbered edges~\footnote{We call an edge $(i,j)$ of a line graph even if $i=0 \mod 2$ and odd otherwise.}, respectively, is provably optimal, see Appendix~\ref{sec:swap_strategy_details}.
This strategy requires $L_S=n-2$ swap layers and is illustrated in Figs.~\ref{fig:transpilation_example} and~\ref{fig:line_strategy}.
For this strategy the minimum number of swap layers needed to reach a density $D$ is $L_S(D)=(n-2)D$.

\begin{figure}
\centering
\includegraphics[width=0.8\columnwidth]{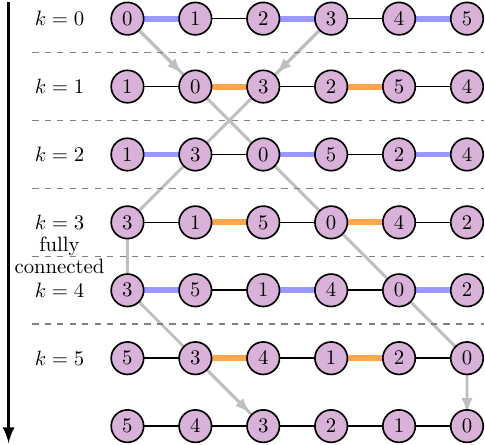}
\caption{
Optimal line swap strategy for $n=6$.
Swap layers are alternately applied to even (in blue) and odd (in orange) numbered edges.
The numbers and gray lines show the positions of logical qubits as they move through the line.
Full connectivity is reached after $k=n-2=4$ layers and the line is fully reversed after $k=n=6$ layers.}\label{fig:line_strategy}
\end{figure}

\subsubsection{Grid coupling map}
We adapt the line graph strategy to the two- and three-dimensional nearest-neighbour grid coupling maps to create strategies that reach full connectivity after $n/2 + \mathcal{O}(\sqrt{n})$ and $n/4 + \mathcal{O}(n^{2/3})$ swap layers, respectively.
For the two dimensional case we consider square grids with $x$ rows and columns, i.e. $n=x^2$.
The swap strategy has four layers $S_0,...,S_3$ and repeats two steps until full connectivity is reached:
\begin{enumerate}
    \item Apply $x-1$ steps of the line swap strategy to each row.
    Importantly, in the same swap layer, the SWAP gates in one of two neighboring rows are applied on even edges while in the other row they are applied on odd edges, see $S_0$ and $S_1$ in Fig.~\ref{fig:grid_strategy}. 
    \item Swap rows by applying two steps of the line swap strategy to each column in parallel, see $S_2$ and $S_3$ in Fig.~\ref{fig:grid_strategy}. Applying $S_2$ and $S_3$ directly after each other guarantees that each row has, up to edge effects, two new neighboring rows, one above and one below itself.
\end{enumerate}
Step 1 creates full connectivity within each row and between adjacent rows. 
Note that while $x-2$ layers are sufficient to reach full connectivity in each row, see Sec.~\ref{sec:line_strat}, we need $x-1$ swap layers to connect all qubits of adjacent rows.
Step 2 swaps rows such that all rows are adjacent to one another at some point in the swap process.
Each iteration of steps 1 and 2 requires $x + 1$ swap layers.
After repeating both steps $\lceil \frac{x}{2} \rceil$ times,  full connectivity is reached. 
In total, the number of swap layers is
\begin{displaymath}
\left\lceil \frac{x}{2} \right\rceil (x+1) \leq \frac{n}{2} + \sqrt{n}+\frac{1}{2}.
\end{displaymath}
This strategy generalizes to grids of higher dimension: for $\eta$-dimensional grid coupling maps a problem density $D$ requires at least $L_S(D)=nD/2^{\eta-1}+\mathcal{O}(n^{1-1/\eta})$ swap layers.
Details are
 in Appendix~\ref{sec:swap_strategy_details}.

\begin{figure}[tb]
 \centering
 \includegraphics[width=0.85\columnwidth]{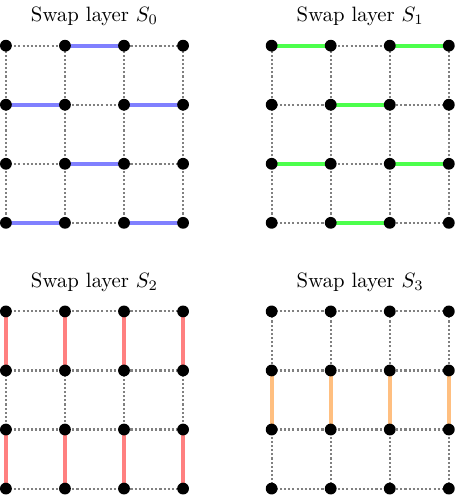}
    \caption{Swap strategy for a $4 \times 4 $  grid.
    The black dots are physical qubits and dashed lines indicate the coupling map.
    The solid colored edges indicate the swap layer of the swap strategy $\mathcal{S}=\{S_0, S_1, S_2, S_3\}$.
    $S_0$ and $S_1$ are repeatedly applied to reach full connectivity in each row and between neighbouring rows (step 1).
    Next, swap layers $S_2$ and $S_3$ are each applied once such that each row becomes adjacent to two different rows (step 2).}
    \label{fig:grid_strategy}
\end{figure}

\subsubsection{Heavy-hex coupling map}
IBM Quantum systems have a heavy-hex coupling map, i.e., a regular hexagonal lattice with additional nodes inserted on each edge~\cite{Chamberland2020}.
A simple strategy applies the optimal line graph strategy to the longest continuous line embedded in the heavy-hex graph. 
This strategy cannot reach full connectivity since a single continuous line does not include all qubits.
However, a swap strategy that applies the line strategy to the longest line in the heavy-hex graph and periodically swaps qubits positioned in the line with qubits that are not part of the line reaches full connectivity after at most
\begin{displaymath}
n + \sqrt{n} + 61 = n + \mathcal{O}(\sqrt{n})
\end{displaymath} swap layers.
Details and a proof are given in Appendix~\ref{sec:swap_strategy_details}.
A lower bound on the number of swap layers to implement a problem with density $D$ is almost linear and we approximate it by $nD$, see Fig.~\ref{fig:graph_density_vs_swaps} and Tab.~\ref{tab:depth_gate_count}.

\subsection{Circuit depth and gate count\label{sec:circuit_depth_and_gate_count}}

\begin{table*}
    \centering
    \begin{tabular}{l r r r r} \hline\hline
        Coupling & $L_S$ & $\quad$Number of CNOT & $\quad$Total number of CNOT gates & $\quad$Average number of CNOT\\ 
        map & & layers $L_\text{cx}$ & per $\exp(-i\gamma H_C)$ layer & gates per layer $l_\text{cx}$\\ \hline
        Line & $Dn$ & $3L_S=3Dn$ & $\frac{3}{2}nL_S=\frac{3}{2}Dn^2$ & $n/2$\\ 
        Grid & $\frac{1}{2}Dn$ & $7 L_S=\frac{7}{2}Dn$& $\frac{7}{2}nL_S=\frac{7}{4}Dn^2$ & $n/2$ \\ 
        3D Grid & $\frac{1}{4} Dn$ & $11 L_S=\frac{11}{4}Dn$ & $\frac{11}{2}nL_S=\frac{11}{8}Dn^2$ & $11n/32$ \\
        Heavy-hex & $Dn$ & $9L_S=9Dn$ & $\frac{8}{5}nL_S=\frac{8}{5}Dn^2$ &  $8n/45$\\ \hline\hline
    \end{tabular}
    \caption{Swap layer and gate count needed to run a depth-one QAOA circuit.
    Here, we only report the leading terms and omit the big-$\mathcal{O}$ for brevity. For $D<1$ these numbers are lower bounds with $D = 1$ being the upper bound. 
    $L_S$ shows the lower bound of the number of swap layers as a function of decision variables $n$ and graph density $D$.
    The derivation of the total number of CNOT gates for a single QAOA circuit is discussed in Appendix~\ref{sec:circuit_depth_and_gate_count_details}.}\label{tab:depth_gate_count}
\end{table*}

\begin{figure}[tb]
\centering
\includegraphics[width=\columnwidth]{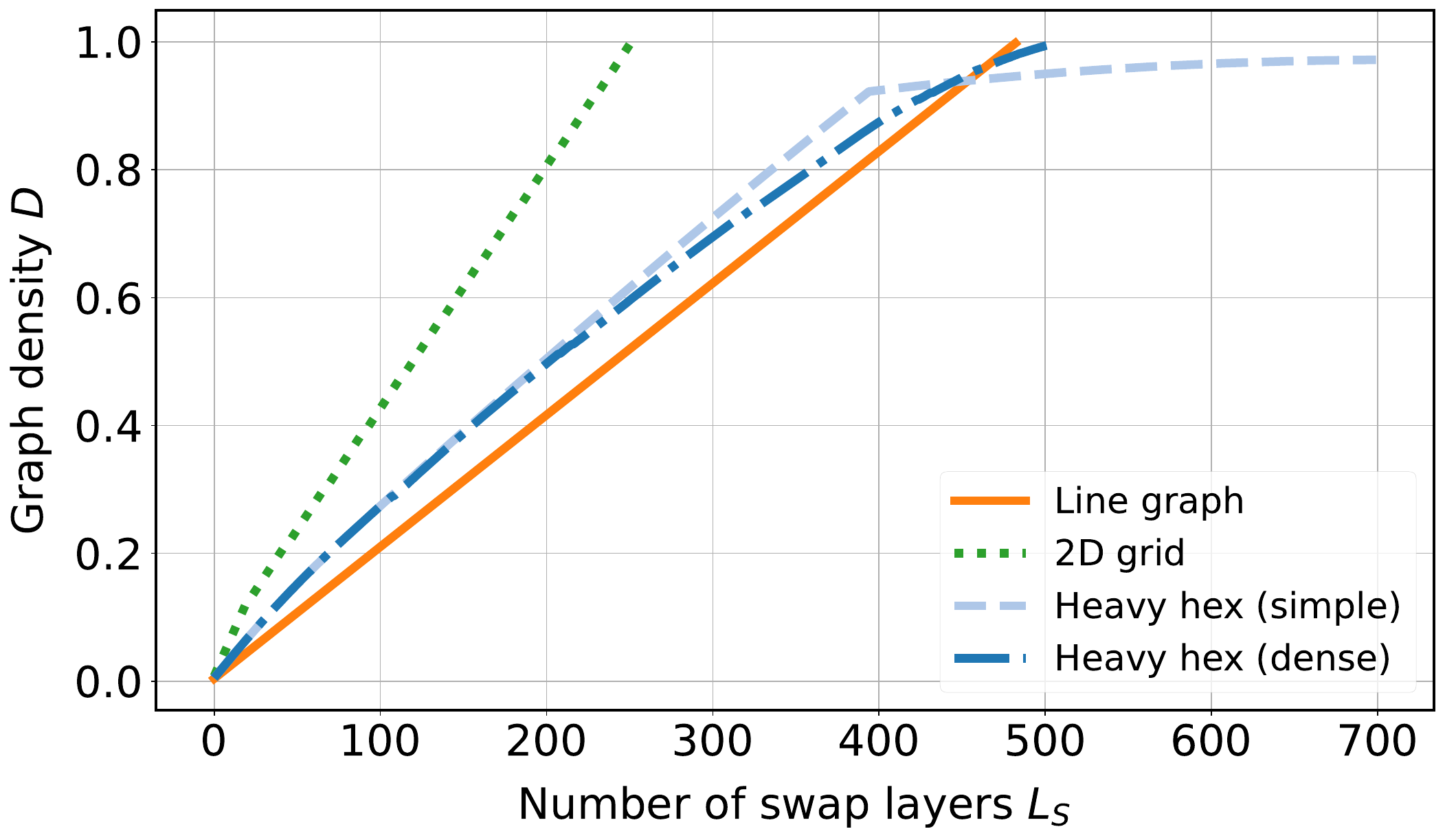}
\caption{Maximal graph density as a function of the number of swap layers $L_S$ for the line graph with 485 qubits (solid orange line), the 2D grid with $22 \times 22=484$ qubits (dotted green line), and the heavy-hex grid with 485 qubits (blue lines).
Here, we chose 484 and 485 qubits so that the different coupling maps have a similar qubit count.
The ''heavy-hex (simple)'' uses the line strategy on the longest line in the heavy-hex graph, and thus, does not reach full connectivity.
}\label{fig:graph_density_vs_swaps}
\end{figure}

We investigate how the gate count and depth of QAOA circuits scale with the swap strategies described in Sec.~\ref{sec:swap_strategies}.
QAOA circuits transpiled with the swap strategies of Sec.~\ref{sec:swap_strategies} consist of alternating layers $E_i$ and $S_i$, containing only $R_{ZZ}$ and SWAP gates, respectively.
The total number of CNOT layers $L_\text{cx}$ and gates is therefore directly related to the required number of swap layers $L_S$.
By definition, every swap layer $S_i$ consists of SWAP gates simultaneously executable on the underlying coupling map. On cross-resonance based hardware a SWAP gate is executed by three CNOT gates. 
The number of CNOT layers in $E_i$ depends on the coupling map, in particular its chromatic edge number, and the swap strategy.
The total number of CNOT gates as a function of $L_S$ is bounded from above by $(4\eta-1)nL_S/2$ for $\eta$-dimensional grids and $8 nL_S/5$ for heavy-hex coupling maps. 
These formulas, derived in Appendix~\ref{sec:circuit_depth_and_gate_count_details}, account for possible gate cancellations between subsequent layers $E_i$ and $S_i$.
We therefore need $\mathcal{O}(Dn^2)$ CNOT gates per application of $\exp(-i\gamma  H_C)$ with prefactors, 1.5, 1.75, 1.375 and 1.6 for linear, 2D-grid, 3D-grid, and heavy-hex coupling maps, respectively, see Tab.~\ref{tab:depth_gate_count}.
The two-dimensional grid strategy is thus worse than the line strategy for the same number of qubits while the three-dimensional grid is the best.
This results from two competing scalings. 
On one hand, the number of swap layers required to reach full connectivity scales as $1/2^{\eta-1}$ with $\eta$, i.e., higher-dimensional grids require less swap layers.
On the other hand, the number of CNOT gates per swap and $R_{ZZ}$ layer combined scales as $4\eta-1$, i.e., higher dimensional grids have less $R_{ZZ}$ and CNOT cancellations.
Therefore, the additional connectivity of the two-dimensional grid over the line is not useful for reaching full connectivity. On two-dimensional grids it is better to use a line swap strategy on the longest line in the grid instead of the gird swap strategy.
By contrast, three-dimensional grids, such as cold atomic lattices~\cite{Henriet2020}, lead to a smaller gate count than a line and a two-dimensional grid.

We now benchmark the swap strategies of Sec.~\ref{sec:swap_strategies} to t$|\text{Ket}\rangle$~\cite{Sivarajah2020}, 2QAN~\cite{Lao2021}, the QAOA Compiler~\cite{Alam2020}, SabreSwap~\cite{Li2019}, and a commutation aware version of SabreSwap that we implemented and describe in Appendix~\ref{sec:commute_sabre}.
We consider depth-one QAOA circuits of MaxCut~\cite{Barahona1988, Hastad2001}, formally introduced in Appendix~\ref{sec:maxcut}, for graphs chosen uniformly at random from the set of all graphs with $n$ nodes and  $\lceil D\frac{n(n-1)}{2} \rceil$ edges~\footnote{The graphs are generated with the \texttt{gnm\_random\_graph} method from the networkx package.}.
We compute the number of layers of simultaneous CNOT gates, and the CNOT gate count after each transpiler pass for heavy-hex coupling maps with a variable number of rows and columns, with $D \in \{0.25, 0.5, 0.75, 1\}$, details are in Appendix~\ref{sec:swap_strat_benchmark_appendix}.
Commutative aware SabreSwap results in fewer CNOT gates but deeper circuits for problems with $D=0.25$ and $0.5$.
2QAN has the least number of CNOT gates but results in deeper circuits than the swap strategies.
The swap strategies are clearly advantageous for dense graphs, see Fig.~\ref{fig:transpiler_comparison}.
This is expected since the swap strategies are tailored for dense problems and may perform unnecessary SWAP gates on qubits that do not need to be connected in incomplete graphs.
For hardware subject to finite $T_1$ and $T_2$-times shallow circuits are advantageous as idling qubits accumulate errors.
We find that the swap strategies result in the shallowest circuits as shown by the number of CNOT layers in Fig.~\ref{fig:transpiler_comparison}.
Additionally, the time taken to transpile with the swap strategies is significantly lower than the time needed by t$|\text{Ket}\rangle$, 2QAN, the QAOA Compiler, and commutative aware SabreSwap and is comparable to SabreSwap, see Appendix~\ref{sec:swap_strat_benchmark_appendix}.

\begin{figure*}[tb]
    \centering
    \includegraphics[width = 0.98\textwidth]{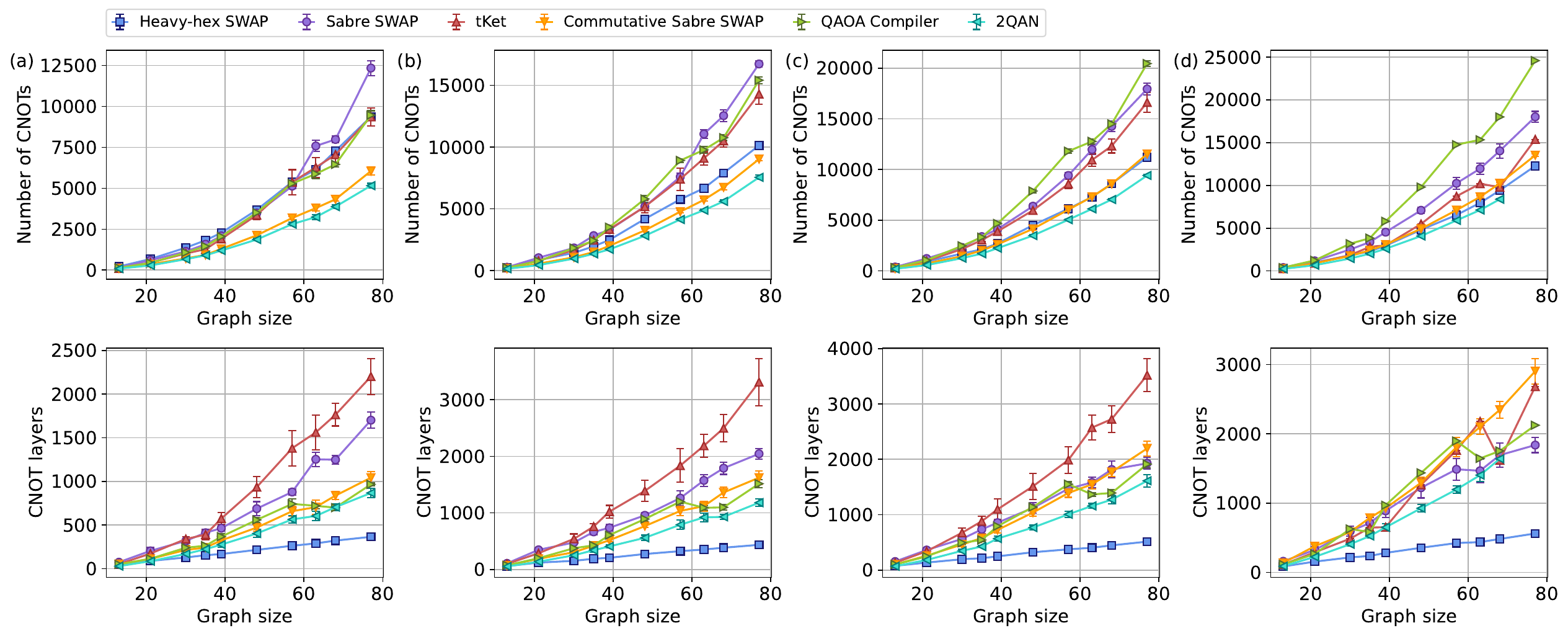}
    \caption{Swap strategies benchmarked on heavy-hex graphs defined by a number of rows and columns $(i, j)\in\{(1, 1), (1, 2), (1, 3), (2, 2), (1, 4), (1, 5), (1, 6), (2, 4), (3, 3), (2, 5)\}$.
    Subfigures (a), (b), (c), and (d) correspond to graph densities of 0.25, 0.5, 0.75 and 1.0, respectively.
    The top row shows the number of CNOTs and the bottom row shows the circuit depth measured as the number of CNOT layers.
    }
    \label{fig:transpiler_comparison}
\end{figure*}

\section{Hardware requirements\label{sec:hardware}}

We now discuss how gate fidelity and gate duration impact QAOA in Sec.~\ref{sec:fidelity} and Sec.~\ref{sec:runtime}, respectively.

\subsection{Gate fidelity\label{sec:fidelity}}

The error-prone unitary gates limit performance.
Entropic inequalities help bound the maximum circuit depth for QAOA~\cite{Aharonov1996, Franca2020}.
Following Proposition 2 of Ref.~\cite{Franca2020}, the maximum depth of a QAOA circuit with a fraction $f_1$ of single-qubit gate layers and a fraction $f_2$ of two-qubit gate layers with depolarizing noise with probability $p_1$ and $p_2$, respectively, is bounded by \begin{align}\label{eqn:depth_criteria}
    L_\text{max}\approx\frac{\ln\epsilon^{-1}}{2(f_1p_1+f_2p_2)}.
\end{align}
For circuits deeper than $L_\text{max}$ there exists a polynomial time classical algorithm that finds a Gibbs state that we can classically sample from with the same energy up to an error $\epsilon\vert\vert H_C\vert\vert$ of the noisy QAOA state.
Here, $\epsilon$ controls the precision with which we approximate the energy.
Reference~\cite{Franca2020} argues that $\epsilon$ should range from $10^{-1}$ to $10^{-2}$ since most optimization algorithms require a number of shots with an inverse polynomial scaling in  $\epsilon$~\cite{Wang2019, Harrow2021}. 
This implies that going beyond an $\epsilon\sim10^{-2}$ incurs a significant sampling cost.
Since the CNOT gate is the dominant source of error and QAOA circuits for denser problems are dominated by two-qubit gate layers, we further simplify Eq.~(\ref{eqn:depth_criteria}) to $L_\text{max}\approx \ln(\epsilon^{-1})/2p_2$.

The CNOT fidelity $\mathcal{F}_\text{cx}$ quoted by IBM Quantum systems is the probability of a depolarizing channel since it is measured with randomized benchmarking~\cite{Magesan2011, Magesan2012, Corcoles2013}.
Each CNOT gate layer in a QAOA circuit transpiled with the swap strategies presented in Sec.~\ref{sec:swap_strategies} will on average have $l_\text{cx}$ CNOT gates, see Tab.~\ref{tab:depth_gate_count}.
We make the simplifying assumption that the depolarizing probability of a layer of CNOT gates is $p_2=1-\mathcal{F}_\text{cx}^{l_\text{cx}}$.
This is an optimistic assumption since effects such as crosstalk may degrade the performance of gates applied in parallel~\cite{Erhard2019, Prakash2020}.
A QAOA with depth $p$ using $L_\text{cx}(n,D)$ CNOT layers per application of $\exp(-i\gamma  H_C)$ for a graph with $n$ nodes and density $D$ must satisfy
\begin{align}\label{eqn:fidelity_criterion}
    pL_\text{cx}(n, D)\leq \frac{\ln\epsilon^{-1}}{2\left(1-\mathcal{F}_\text{cx}^{l_\text{cx}}\right)},
\end{align}
otherwise there is a corresponding Gibbs state which can be sampled from classically in polynomial time~\cite{Franca2020}.
Following Ref.~\cite{Franca2020} there is therefore little chance that running a quantum circuit that requires more CNOT layers than the bound in Eq.~(\ref{eqn:fidelity_criterion}) will lead to a quantum advantage, i.e., we assume there is little hope for a quantum advantage if there exists a polynomial-time classical approximation.

We calculate the bound in Eq.~(\ref{eqn:fidelity_criterion}) for a heavy-hex coupling map with 485 qubits as a function of $D$ and $\mathcal{F}_\text{cx}$.
The density-dependent upper bound on the gate error is one to three orders of magnitude lower than current hardware capabilities, see Fig.~\ref{fig:fidelity_vs_graph_density}.
The data indicate that non-hardware native optimization problems will require gate fidelities above error correction thresholds which typically range from 99\% to 99.99\%~\cite{Aliferis2007, Fowler2012, Chamberland2020}.
When such fidelities are reached it may still be advantageous to run QAOA on noisy hardware due to the large qubit overhead imposed by error correction and the potentially lower execution times.

\begin{figure}[htbp!]
    \centering
    \includegraphics[width=\columnwidth]{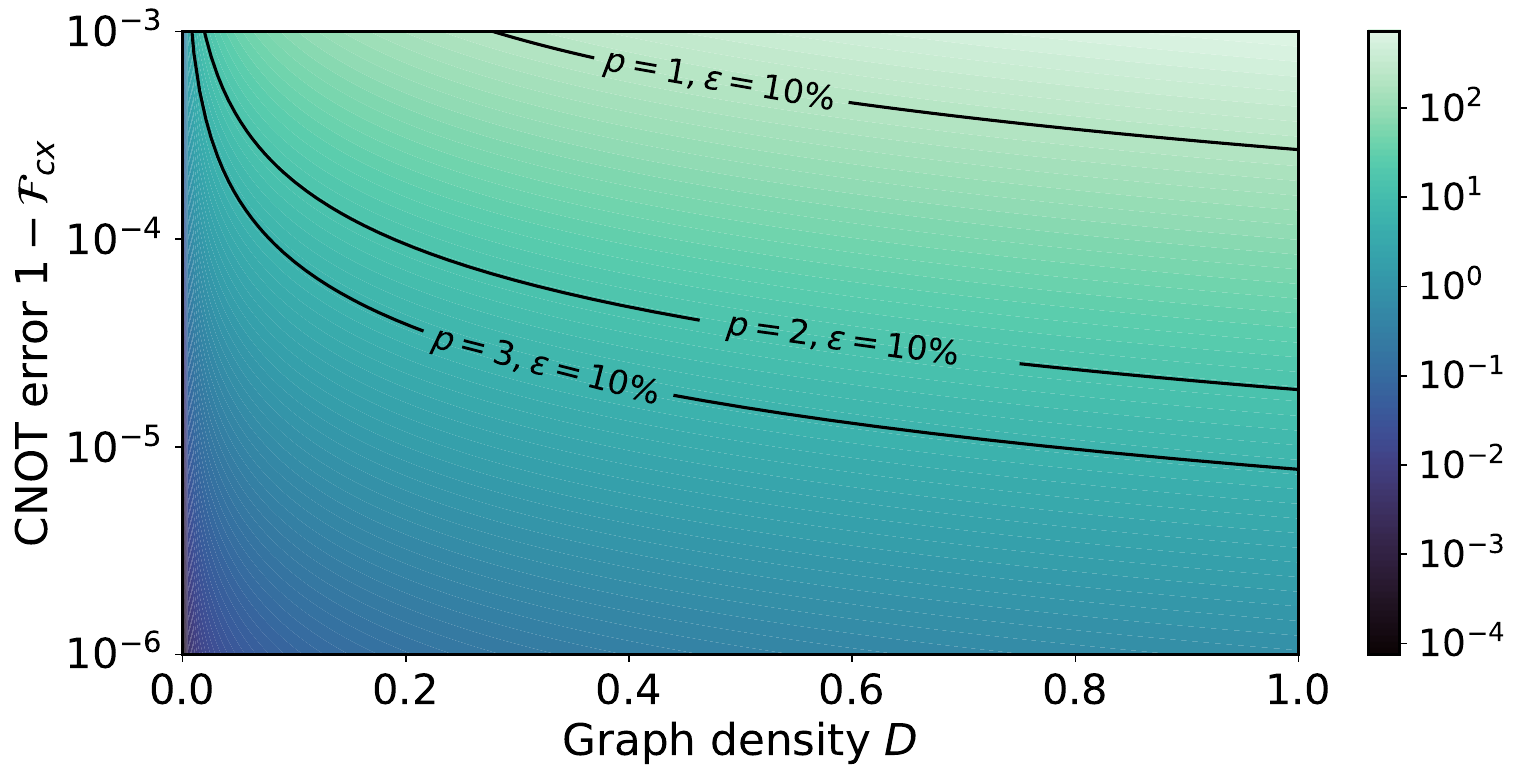}
    \caption{
    Gate depth criterion to implement QAOA on a heavy-hex graph. 
    The color scale shows $18nD(1-\mathcal{F}_\text{cx}^{l_\text{cx}})$ plotted on a logarithmic scale as a function of graph density $D$ and $1-\mathcal{F}_ \text{cx}$ with $n=485$ (as in Fig.~\ref{fig:graph_density_vs_swaps}) for the heavy-hex swap strategy which has an average of $l_\text{cx}=8n/45$ CNOT gates per CNOT layer, see Tab.~\ref{tab:depth_gate_count}.
    The contour lines indicate $\ln(\epsilon^{-1})/p$ with $\epsilon=10\%$ and different QAOA depths $p$.
    }
    \label{fig:fidelity_vs_graph_density}
\end{figure}

\subsection{Execution-time analysis\label{sec:runtime}}

We now estimate the execution-time of QAOA as the product of the number of iterations of the classical solver $N_\text{iter}$ times the time taken to gather the data at each iteration
\begin{equation}
    \tau_\text{QAOA}=N_\text{iter}\left( N_\text{shots}\cdot\tau_\text{shot}+\tau_\text{init}\right).
\end{equation}
Here, the number of measurements per iteration is $N_\text{shots}$ and the duration of a single-shot $\tau_\text{shot}=\tau_\text{circ}+\tau_\text{delay}$.
The time taken to run all the gates, measurement, and reset instruction is captured in $\tau_\text{circ}$ while $\tau_\text{delay}$ is a fix delay after each measurement used to improve the qubit reset~\cite{Wack2021, Egger2018}.
At each iteration the control hardware must be setup to gather the next shots and therefore incurs a time penalty $\tau_\text{init}$.
This decomposition is similar to the Circuit Layer Operations per Second (CLOPS) benchmark~\cite{Wack2021}.
We estimate the execution-time $\tau_\text{QAOA}$ for different problem sizes $n$ and on different coupling maps.
We focus on problems with $D=1$ since they upper bound the $D<1$ instances.
Substituting $D<1$ in the following equations may underestimate the execution time depending on the graph structure.

\begin{table}[]
    \centering
    \begin{tabular}{l r r}\hline\hline
                              & Current & Possible \\ \hline
        $\sqrt{X}$ error     & $7.2\cdot 10^{-4}$ & \\
        CNOT error           &  $169.0\cdot 10^{-4}$ & \\ 
        $\sqrt{X}$ duration     & $36~{\rm ns}$ & $4~{\rm ns}$~\cite{Werninghaus2020}\\
        CNOT duration        & $400~\rm{ns}$        & $30~{\rm ns}$~\cite{Kirchhoff2018} \\
        \hline\hline
    \end{tabular}
    \caption{
    Performance metrics of quantum hardware.
    Current performance corresponds to typical values on IBM Quantum systems.
    Possible performance shows the gate duration that has been measured or could be implemented. \label{tab:hardware_values}}
\end{table}

\subsubsection{Duration of a single-shot}

On cross-resonance based hardware~\cite{Rigetti2010, Sheldon2016} the duration of a single-shot is determined by the duration of the CNOT gate $\tau_\text{cx}$, the QAOA depth $p$, the problem density $D$, and the coupling map as discussed in Sec.~\ref{sec:swap_strategies}.
Since QAOA at sub-logarithmic depth is not expected to outperform classical solvers~\cite{Hastings2019, Bravyi2020, Farhi2020}, we assume that $p$ scales at least logarithmically with the number of variables $n$. 
We therefore chose $p=\log_2(n)$ for our runtime analysis.
With the number of CNOT layers $L_\text{cx}(n,D)$ we estimate that a single-shot lasts at least
\begin{align}
    \tau_\text{shot}\geq\log_2(n)L_\text{cx}(n,D)\tau_\text{cx}+\tau_\text{delay}.
\end{align}
Since $L_\text{cx}(n,D)$ scales as $\Omega(Dn)$ the duration of a single QAOA shot scales at least as fast as $\tau_\text{shot} = \Omega(Dn\log_2 n)$.
With a $400~{\rm ns}$ CNOT gate and a heavy-hex coupling map, i.e. $L_\text{cx}=9nD$, the duration of a shot is significant, see Fig.~\ref{fig:tau_qaoa} which only includes the CNOT gate time.
Here, measurement and reset instructions, which can last up to a few micro-seconds~\cite{Geerlings2013, Egger2018, Corcoles2021} and typically only appear once in a quantum circuit, are neglected.
With current cross-resonance gate durations of $200-400~{\rm ns}$, delays can also be neglected since $\tau_\text{delay}\approx250~{\rm \mu s}$ for current hardware~\cite{Wack2021} which is two orders of magnitude faster than the circuit duration.
The CNOT duration is thus currently the main driver of execution-time on noisy quantum hardware.
For example, the circuit of a complete interaction graph with $485$ variables has a single-shot duration of $14.9~{\rm ms}$ on a heavy-hex coupling map.
Optimal control schemes show that it is in principle possible to reduce the duration of the single- and two-qubit gates by an order of magnitude~\cite{Werninghaus2020, Kirchhoff2018}, see Tab.~\ref{tab:hardware_values}.
This reduces the QAOA run-time by an order of magnitude and will make the fixed delay $\tau_\text{delay}$ after each shot more relevant.

\begin{figure}[tb]
    \includegraphics[width=\columnwidth]{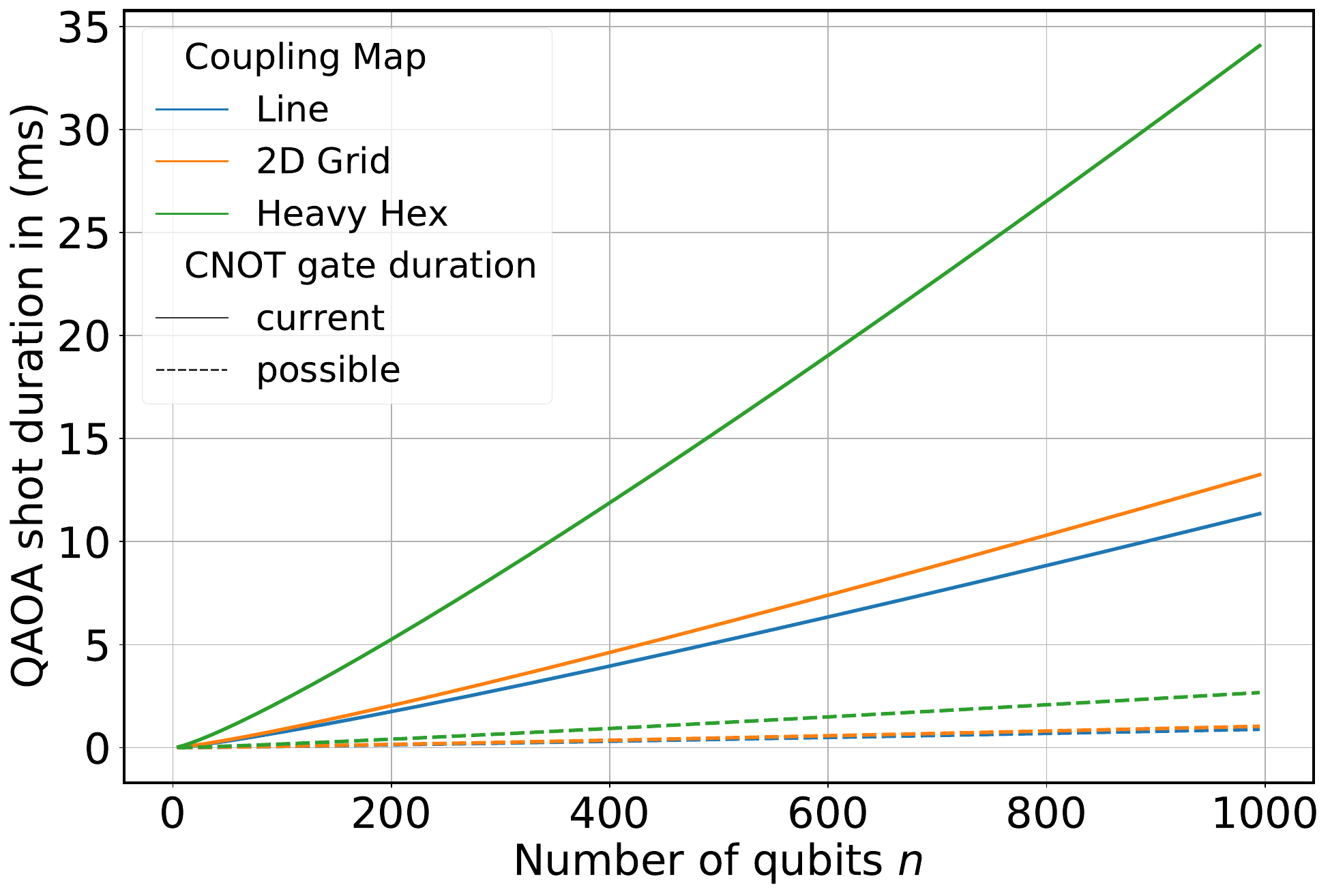}
    \caption{
    Scaling of the duration of a single QAOA circuit $\tau_\text{shot}$ as a function of problem size $n$ for complete graphs with $p=\log_2(n)$.
    The solid lines show current CNOT gate durations while the dashed lines correspond to $\tau_\text{cx}=30~{\rm ns}$.
    }
    \label{fig:tau_qaoa}
\end{figure}

\subsubsection{Number of shots required\label{sec:shots_per_iter}}

The classical optimizer has to minimize the objective function $E(\boldsymbol\theta)=\braket{\psi(\boldsymbol\theta)| H_C|\psi(\boldsymbol\theta)}$, where $\boldsymbol\theta=(\boldsymbol{\beta}, \boldsymbol{\gamma})$, which it can only stochastically access~\cite{Sweke2020}.
A simulation of a variational algorithm must therefore consider this effect. 
For example, QAOAs simulated in Qiskit~\cite{Qiskit} always perform better with the state-vector simulator, which does not have sampling noise, than the shot-based QASM simulator, see Fig.~\ref{fig:n_shots}.
A large number of shots thus helps QAOA converge~\cite{Zhou2020, Sankar2021} but increases its execution-time.
Zeroth-order methods optimize by directly estimating $E(\boldsymbol\theta)$~\cite{Harrow2021}. 
They prepare and measure 
each trial state $N_\text{shots}$ times.
Measurement $k$ randomly projects $\ket{\psi(\boldsymbol{\theta})}$ onto a basis-state with an energy $E_k$ thereby estimating $E(\boldsymbol\theta)$ by $\overline E=N_\text{shots}^{-1}\sum_k E_k$.
Reference~\cite{Harrow2021} shows that for 1-local Hamiltonians with $n$ qubits the total number of shots over the course of the optimization, i.e., $N_\text{iter}N_\text{shots}$, needed to reach a precision $\epsilon$ within the vicinity of the optima is lower bounded by $\Omega(n^3/\epsilon^2)$.
By contrast, first-order methods optimize by taking measurements that correspond to the gradient $\partial_{\boldsymbol{\theta}}E(\boldsymbol{\theta})$ and are sometimes referred to as analytical gradient measurements~\cite{Guerreschi2017, Romero2018}.
Furthermore, for 1-local Hamiltonians Ref.~\cite{Harrow2021} shows that the total number of shots required by first-order methods to reach an $\epsilon$ precision scale as $\Theta(n^2/\epsilon)$.
In this setting, first-order methods therefore converge faster but still require a large number of shots.
Recently, optimizers that scale the number of shots based on the magnitude of the gradient have been developed to reduce the shot cost~\cite{Kubler2020}.
The $\Theta(n^2/\epsilon)$ scaling is a significant shot cost.
However, since optimal QAOA parameters concentrate for similar problem instances~\cite{Brandao2018} this cost may be amortized when solving many problem instances originating from the same reasonable distribution.

\begin{figure}[tb]
 \includegraphics[width=\columnwidth]{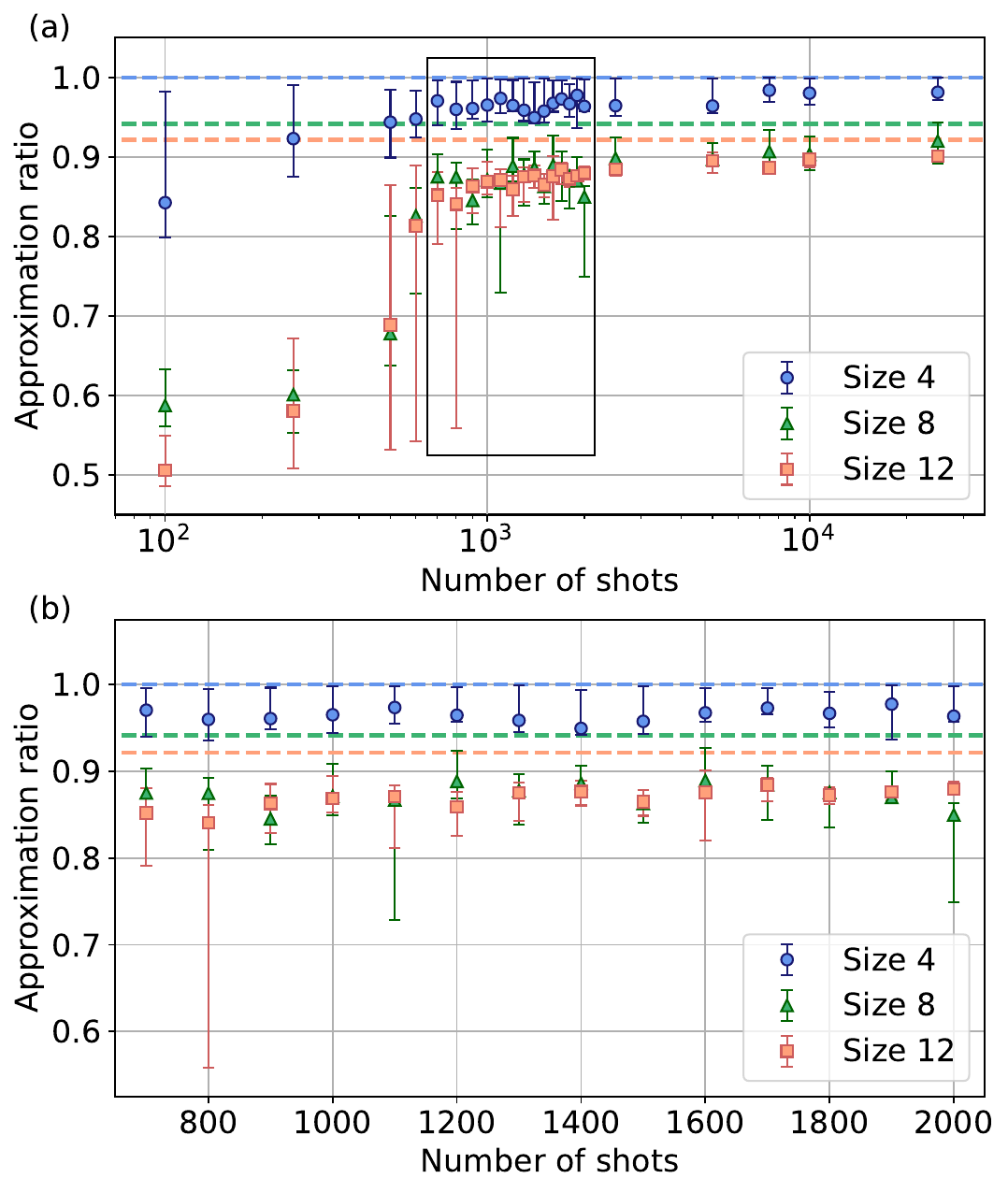}
 \caption{
 Depth 12 QAOA for Sherrington-Kirkpatrick graphs with 4, 8 and 12 nodes optimized with COBYLA.
 (a) Energy normalized to the ground-state energy, i.e., approximation ratio, after a depth 12 QAOA optimization. 
 The points show the median of ten different random instances as a function of the number of measured shots.
 The error bars show the 25\% and 75\% quantiles.
 The dashed lines show a state vector simulation which is equivalent to $N_\text{shots}\to\infty$ up to machine precision. 
 (b) Data points within framed region in (a).
 \label{fig:n_shots}
 }
\end{figure}

\begin{figure}[tb]
    \centering
    \includegraphics[width=\columnwidth]{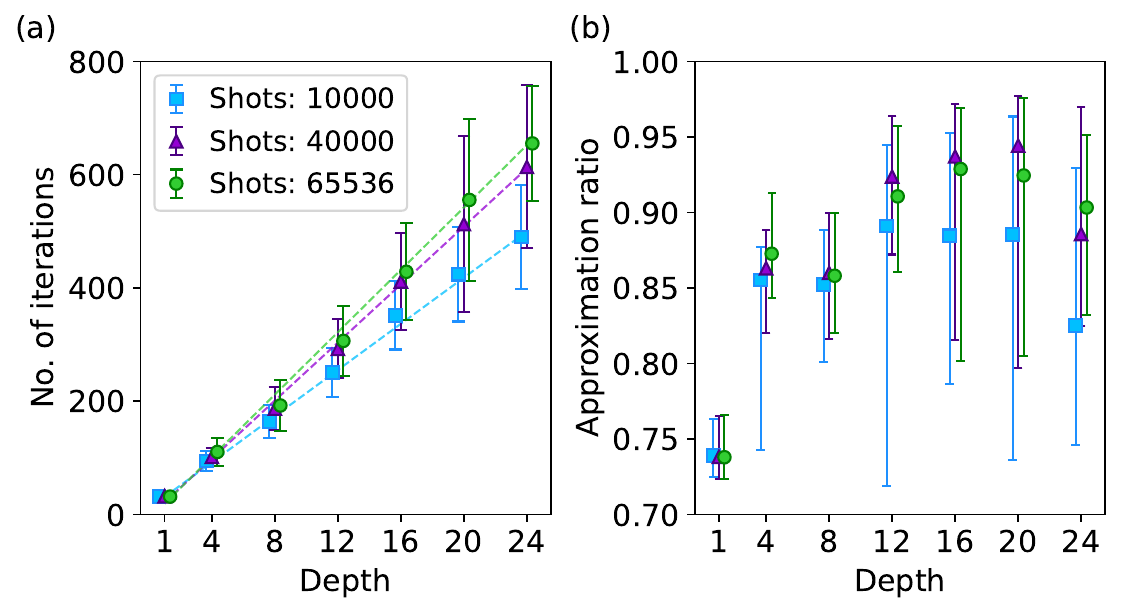}
    \caption{(a) Number of iterations completed by COBYLA.
    Each data point is the average number of iterations for 30 graphs with the standard deviation indicated as error bars. A small $p$ shift is added for visualization purposes.
    The number of iterations is approximately constant with $n$ (data not shown) and grows linearly with $p$, as shown by the linear fits (dashed lines) with slops of 20.4, 25.6, and 27.5 for $10^4$, $4\cdot 10^4$, and $2^{16}$ shots, respectively.
    (b) The approximation ratio increases with QAOA depth.
    A larger number of shots results in better approximation ratios and more COBYLA iterations.}
    \label{fig:cobyla_iter}
\end{figure}

\subsubsection{Number of iterations}

We further investigate the number of iterations of the classical optimizer empirically with COBYLA by running QAOA simulations on Qiskit's QASM simulator.
We set the number of maximum iterations to a high value (100,000) so that COBYLA terminates before reaching this threshold. 
The number of completed iterations is recorded for Sherrington-Kirkpatrick graphs with size $n$ from 5 to 10, with edge weights randomly chosen from $\{-1,1\}$ resulting in a total of 30 graphs, i.e., five per size.
We initialize $\boldsymbol{\beta}$ and $\boldsymbol{\gamma}$ for QAOA depths $p\in\{1,4,8,12,16,20,24\}$ with the Trotterized Quantum Annealing (TQA) protocol, i.e. a discretized annealing schedule, with a time-step of 0.75~\cite{Sack2021} as it performs better than random guesses.
We observe that the number of iterations grows linearly with $p$ for all simulated graphs, see Fig.~\ref{fig:cobyla_iter}.
At fixed $p$ we did not see a change due to the graph size.
The mean cut value shows a noticeable improvement when increasing the shots from $10^4$ to $4\cdot10^4$ and the number of iterations increases with the number of shots since the optimizer is able to resolve finer details in the optimization landscape.
Since we use $p=\Omega(\log_2 n)$ we therefore approximate $N_\text{iter}\approx 25\log_2 (n)$.
This estimate is obtained with noiseless simulations.
In practice experimental noise will make it harder to converge to a good solution~\cite{Quiroz2021} and advances in optimizers for variational quantum algorithms may help speed-up convergence~\cite{Gacon2021}.
Furthermore, if $N_\text{iter}N_\text{shots}$ scales as $\Theta(n^2/\epsilon)$, then a logarithmic scaling of $N_\text{iter}$ with $n$ implies that many shots are needed at each iteration further confirming that hardware initialization times can be neglected, i.e., $N_\text{shots}\tau_\text{shot}\gg \tau_\text{init}$.

\subsubsection{Total QAOA execution time}

We now combine the results from the previous three sections to estimate the execution time of QAOA as
\begin{align*}
    \tau_{QAOA}&\approx N_\text{iter}N_\text{shots} \cdot\log_2(n)L_\text{cx}(n,D)\tau_\text{cx}.
\end{align*}
Here, the linear dependence on $p$ of the single-shot duration gives a factor $\log_2 (n)$.
The required number of shots is the largest source of uncertainty in the estimation of the execution time. 
We require at least $N_\text{shots}(n)\geq 10^3$ for small problems in noiseless conditions.
If $N_\text{iter}N_\text{shots}$ scales as $\Theta(n^2/\epsilon)$, as suggested by 1-local Hamiltonians~\cite{Harrow2021}, the impact on the execution time will be significant even for problems with a few hundred of variables, see Fig.~\ref{fig:total_run_time}.
With $10^4$ shots we estimate that the execution time of a complete graph with 485 nodes on a heavy-hex processor is $9.7~{\rm hours}$, see Fig.~\ref{fig:total_run_time}.
With the same number of shots a sparse graph with $D=0.1$ would require at least one hour to execute.
Furthermore, these estimates show that decreasing the CNOT gate duration is crucial.

We have not taken into account the cost of error mitigation strategies.
For example, the expectation value can be extrapolated to the zero-noise level by measuring the energy at different noise levels~\cite{Kandala2018, Temme2017}.
This multiplies $N_\text{shots}\cdot\tau_\text{shot}+\tau_\text{init}$ by the number of noise-levels measured.
Less noisy energy evaluations may also reduce the number of iterations needed.

Care must be taken when comparing quantum-based optimizers to classical solvers~\cite{Dash2015} but with current quantum technology the estimated execution appears to be significant~\cite{Billionnet2007}.
However, a possible advantage of QAOA over classical optimizers could lie in quickly generating good yet sub-optimal solutions by foregoing the classical optimization algorithm and initializing $\boldsymbol{\beta}$ and $\boldsymbol{\gamma}$ from a known good schedule, such as an annealing schedule~\cite{Sack2021}.
Furthermore, as optimal QAOA parameters tend to concentrate, optimizing one problem may be sufficient for a family of similar problems \cite{Brandao2018, Akshay2021, Galda2021, Streif2020, Shaydulin2022}.
We account for this possibility by removing the factor $N_\text{iter}=25\log_2(n)$ in the execution time and use a fixed number of shots $N_\text{shots}=10^4$.
Under these assumptions candidate solutions for a graph with 500 nodes can be generated in under three minutes, see the dashed-dotted line in Fig.~\ref{fig:total_run_time}.

\begin{figure}[tb]
    \centering
    \includegraphics[width=\columnwidth]{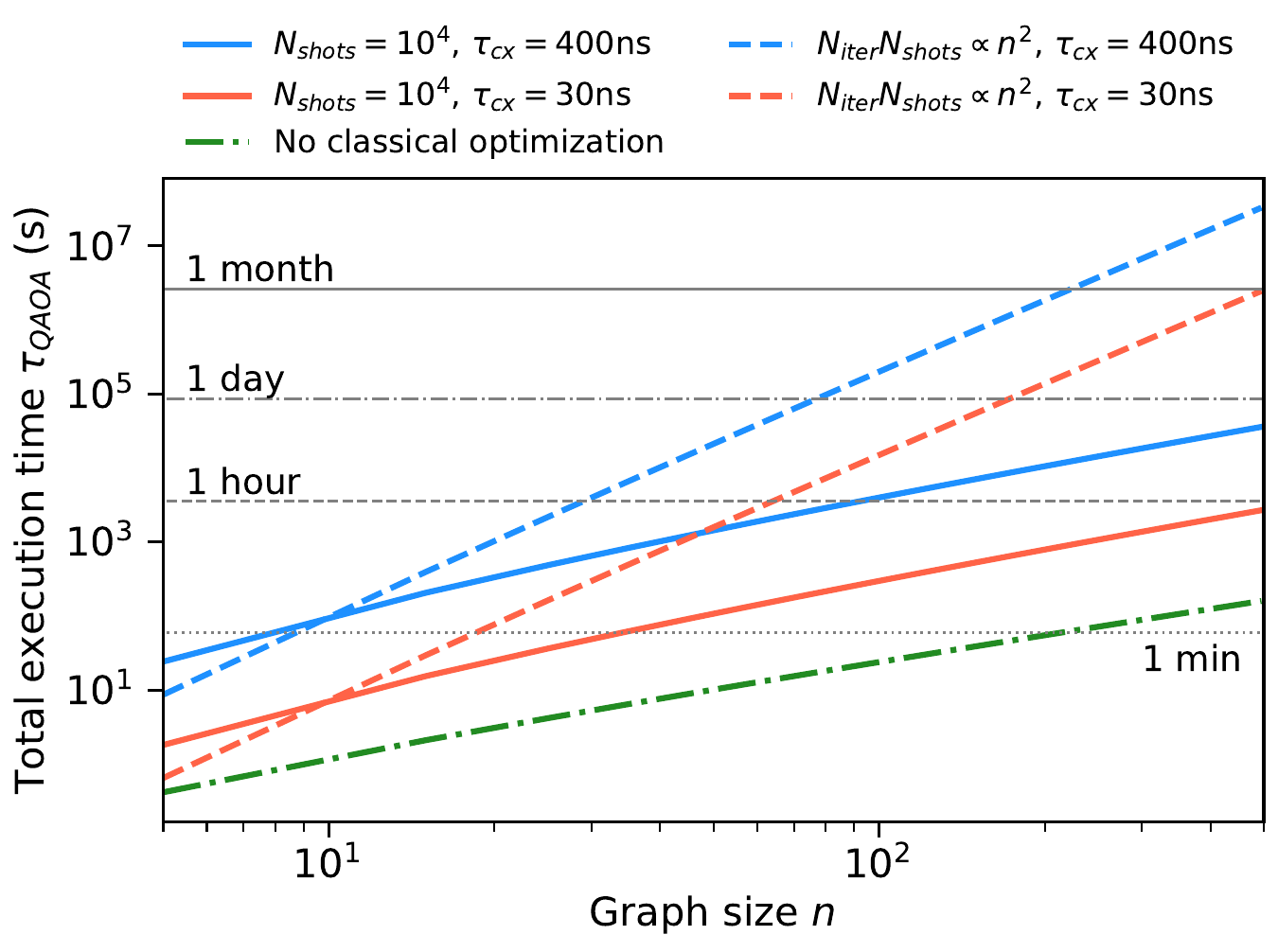}
    \caption{
    Total execution-time for complete graphs computed as $N_\text{iter} N_\text{shots}\tau_\text{shot}$ on a heavy-hex coupling map, i.e. $\tau_\text{shot}=L_{CX}(n)\log_2(n)\tau_\text{cx}$ with $L_{CX}(n)=9n$.
    The quadratic scaling (dashed lines) corresponds to $N_\text{iter}N_\text{shots}=25\cdot 10^4\log_2(10) (n/10)^2$ chosen such that a size $n=10$ graph requires a total of $25\log_2(10)\cdot 10^4$ shots following
     Fig.~\ref{fig:cobyla_iter}.
    Solid lines show a fixed number of shots per iteration, i.e., $N_\text{iter}N_\text{shots}=25\cdot10^4\log_2(n)$.
    The dashed-dotted line shows the total execution time without classical optimization, i.e. $N_\text{iter}=1$, with a fixed number of shots $N_\text{shots}=10^4$ and CNOT gates lasting $400~\rm{ns}$. This corresponds to a case where the classical optimization has been amortized over a large number of problem instances.
    }
    \label{fig:total_run_time}
\end{figure}

\begin{table*}
    \centering
    \begin{tabular}{l r r r r c r r r r} \hline\hline
        & \multicolumn{4}{c}{CNOT count} & $\quad$ & \multicolumn{4}{c}{Schedule duration ($\mu{\rm s}$)} \\
        & \multicolumn{4}{c}{depth $p$} & & \multicolumn{4}{c}{depth $p$}\\  
        $\mathcal{G}_{10}$ & 1 & 2 & 3 & 4 & & 1 & 2 & 3 & 4\\ \hline
        Unoptimized & $~$39 & $~$89 & $~$127 & $~$173 & & $~$6.6 & $~$14.6 & $~$21.2 & $~$28.6 \\
        Optimized & 24 & 48 & 72 & 96 & & 2.1 & 4.1 & 6.1 & 8.1 \\ 
        \hline\hline
    \end{tabular}
    \caption{CNOT gate count and schedule duration for $\mathcal{G}_{10}$.
    The unoptimized row corresponds to the default transpiler settings in Qiskit while the optimized row corresponds to the swap strategies followed by a pulse-efficient transpilation shown for $\gamma=\pi/4$.}
    \label{tab:circuit_metrics}
\end{table*}

\section{Qiskit Runtime hardware results \label{sec:runtime_program}}

The current execution model on cloud-based quantum computers sends a set of circuits as a job through the entire stack and queue.
Circuit transpilation and result analysis is done on the client side.
This is particularly inefficient for variational algorithms.
The Qiskit Runtime allows users to run an entire program in a containerized service close to the backend to avoid latencies between the user and the backend at each iteration.
This enables a significantly faster execution of variational algorithms like QAOA~\footnote{The QAOA Runtime program is publicly available through the IBM Quantum Services \href{https://quantum-computing.ibm.com/}{https://quantum-computing.ibm.com/}.}.

We first demonstrate the QAOA Runtime program with a seven-variable weighted maximum cut  optimization problem with a graph $\mathcal{G}_{10}$ with 10 unique Pauli $Z_iZ_j$ terms and depths $p\in\{2,3,4\}$.
Each edge $(i,j)$ has a weight $\omega_{i,j}$ of $-1$ or $1$ with a $50\%$ probability, see Appendix~\ref{sec:graph_details}.
$\mathcal{G}_{10}$ was constructed such that it can be implemented with one swap layer on the seven-qubit \emph{ibm\_nairobi} system.
Since the energy $E$ is related to the cut value $C$ by $E=-2C+\sum_{(i,j)}\omega_{ij}$ we minimize the energy to maximize the cut.
For $\mathcal{G}_{10}$ we have $E=-2C$.
We run QAOA with SPSA due to the noisy environment~\cite{Spall1998, Spall1997} and measure $2^{14}$ shots at each energy evaluation.
The 0.005 learning rate and 0.01 perturbation of SPSA were chosen by calibrating it on a depth-one landscape.
The initial $\boldsymbol{\gamma}$ and $\boldsymbol{\beta}$ values come from TQA initialization~\cite{Sack2021}.
First, we run $\mathcal{G}_{10}$ three times with $p=2$ with the default transpiler setting of Qiskit which does not use swap strategies.
Here, we do not observe any convergence, see the red data in Fig.~\ref{fig:seven_qubit_nairobi}.
We attribute this to the 0.83\% average CNOT gate error of \emph{ibm\_nairobi} and the deep CNOT gate count of the circuits.
Indeed, the circuits have 89, 127, and 173 CNOT gates for depth $p=2,3$, and $4$, respectively.

\begin{figure}[tb!]
    \centering
    \includegraphics[width=\columnwidth]{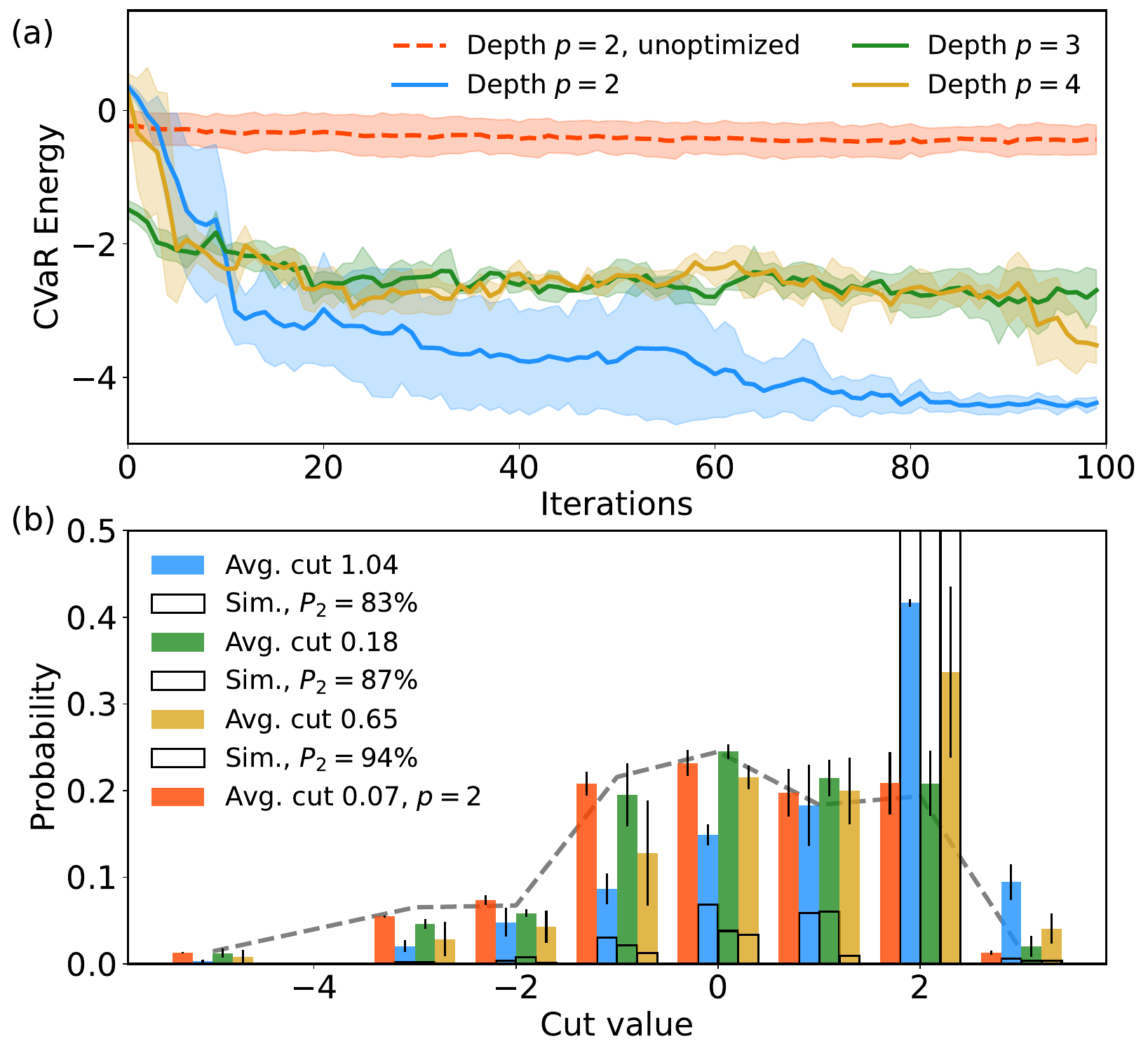}
    \caption{
    Execution of a seven-qubit QAOA on \emph{ibm\_nairobi}.
    (a) The average CVaR energy of SPSA's two evaluations per iteration averaged over three independent optimizations with TQA initial points.
    The shaded area is the standard deviation of the three runs.
    The dashed red line shows a depth-two QAOA with unoptimized transpiler settings.
    Solid lines show QAOA transpiled with swap strategies, pulse-efficient, and executed with CVaR with $\alpha=0.5$.
    (b) Cut distribution of the best measured point averaged over the three optimizations.
    Vertical black lines show the standard deviation.
    Here, the dashed gray line indicates the distribution obtained by uniformly sampling integers from $\{0, \ldots, 2^n-1\}$ and converting them to a bit string representing a cut.
    The black empty bars show the histogram of a noiseless QASM simulation with $2^{14}$ shots.
    The legend contains the simulated probability $P_2$ of a cut with value two.
    The execution time of each job is on average 1 hour and 5 minutes.}
    \label{fig:seven_qubit_nairobi}
\end{figure}

To improve convergence we rerun the QAOA Runtime program with optimized settings.
First, we use the swap strategies discussed previously.
Second, after the parameters are bound we run a pulse-efficient transpiler pass~\cite{Earnest2021} at each QAOA iteration to remove any unnecessary single-qubit gates and to minimize the cross-resonance gate usage.
The gains from these transpiler passes are summarized in Tab.~\ref{tab:circuit_metrics}.
Third, we employ CVaR optimization with an $\alpha$ of 0.5, i.e. we retain only the best 50\% of the measured shots at each iteration~\cite{Barkoutsos2020}.
Finally, we use readout error mitigation to reduce measurement errors~\cite{Bravyi2020b, Barron2020a}.
We observe a significant reduction in energy for $\mathcal{G}_{10}$ for depth $p=2$ as function of the iteration number.
Depths 3 and 4 would require 72 and 96 CNOT gates, respectively, without pulse-efficient transpilation and are most likely noise limited.
Crucially, the jobs that did see convergence manage to increase the probability of sampling a good cut when compared to random sampling, compare the histograms in Fig.~\ref{fig:seven_qubit_nairobi} with the dashed grey line.
Interestingly, the hardware has a higher probability of sampling the maximum cut which may be due to noise such as $T_1$-induced errors.
Crucially, each run of the variational algorithm required only one hour on the cloud-based quantum computer.
In addition, we evaluate the criterion in Eq.~(\ref{eqn:fidelity_criterion}).
A single layer of $ H_C$ has eleven and two layers of two and one CNOT gates, respectively.
We therefore approximate this as 12 layers with two gates.
The average gate fidelity on \emph{ibm\_nairobi} is $98.88\%$ which results in a maximum gate bound of 52 and 103 layers for $\epsilon$ values of $10^{-1}$ and $10^{-2}$, respectively.
Equivalently, for the depth-four QAOA, which has 48 layers, there is a Gibbs state which can be sampled from classically in polynomial time which approximates the energy to within an $\epsilon$ of $0.115$, see Sec.~\ref{sec:fidelity}.

\begin{figure}[tb!]
\includegraphics[width=\columnwidth]{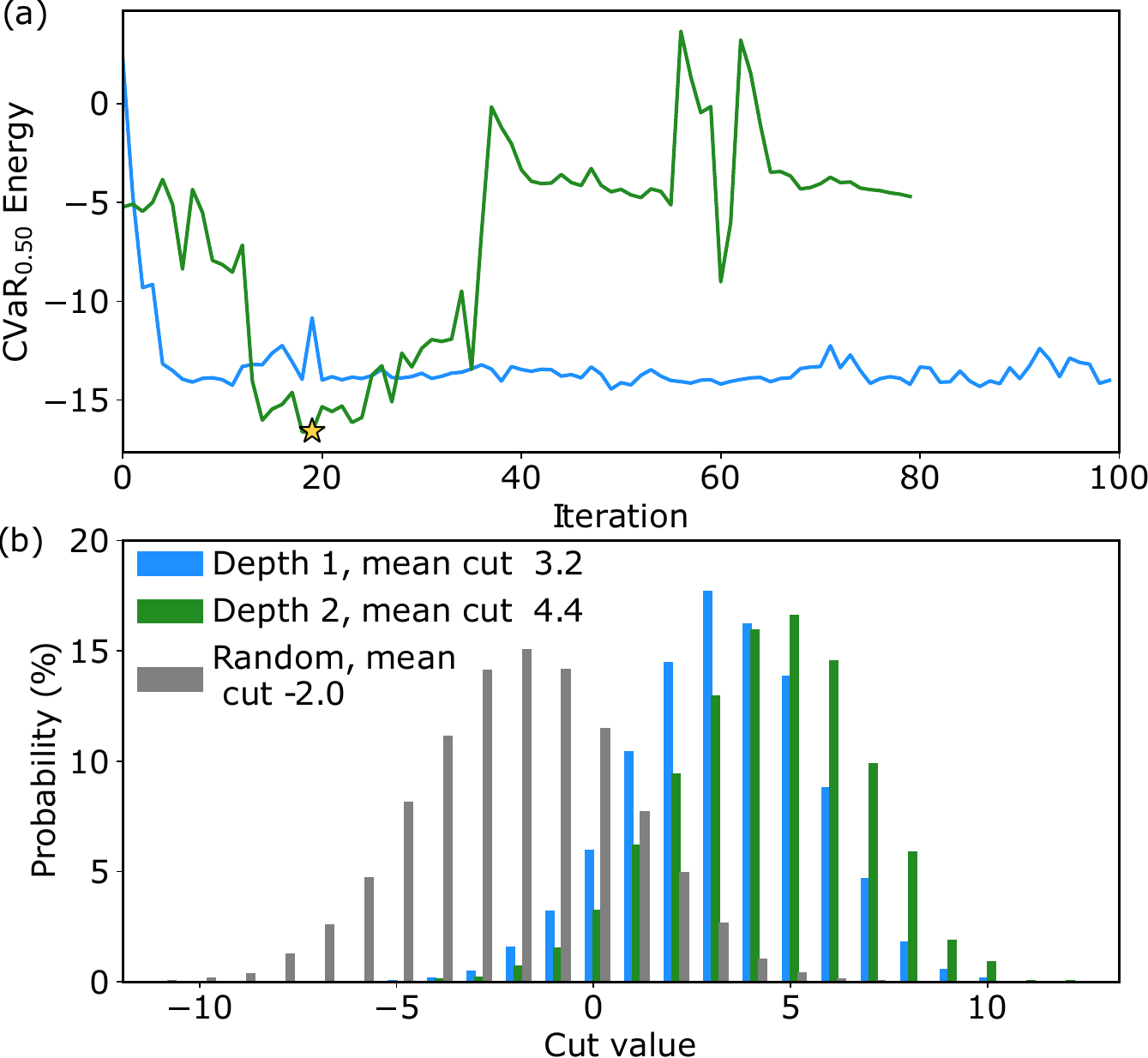}
\caption{
QAOA on a weighted maximum cut problem with a graph native to the hardware of \emph{ibmq\_mumbai} with a maximum cut of 12.
At each energy evaluation 16384 shots are gathered.
(a) Energy as function of the iteration number for depth-one and two QAOA initialized at $\gamma=1$ and $\beta=0.5$ and $\boldsymbol{\gamma}=(0.375, 0.75)$ and $\boldsymbol{\beta}=(0.375, 0)$, respectively.
The star with a ${\rm CVaR}_{0.5}$ energy of $-16.57\pm 0.16$ marks the depth-two parameters $\boldsymbol{\gamma}=(1.1245, 0.343)$ and $\boldsymbol{\beta}=(1.317, 0.548)$ that resulted in the green distribution of cuts.
(b) Probability distribution of the best measured cuts for depth-one (blue), depth-two (green), and random sampling (gray).
The energy $E$ is related to the cut value $C$ by $E=-2C-4$.
\label{fig:mumbai}
}
\end{figure}

\begin{table*}
\centering{
\small
\begin{tabular}{l r r r r r r}\hline\hline
Cut         & & \multicolumn{2}{c}{QAOA} & \multicolumn{3}{c}{Single-shot probability} \\
size & Graph & $p=1$ & $p=2$ & $\quad$Uniform & $p=1$ & $p=2$\\ \hline
10 & 212 & 25 & 148 & $1.6\cdot 10^{-6}$ & $\;\;1.5\cdot 10^{-3}$ & $\;\;9.0\cdot 10^{-3}$ \\
11 &  12 &  3 &   6 & $8.9\cdot 10^{-8}$ & $1.8\cdot 10^{-4}$ & $3.7\cdot 10^{-4}$ \\
12 &   2 &  1 &   7 & $1.5\cdot 10^{-8}$ & $6.1\cdot 10^{-5}$ & $4.3\cdot 10^{-4}$\\\hline\hline
\end{tabular}
\caption{
Number of cuts with size $\geq 10$.
Graph: number of cuts of each size in the Graph.
QAOA: number of measured cuts out of the $2^{14}$ shots shown in Fig.~\ref{fig:mumbai}(b).
The last three columns indicate the probability of a single-shot producing a cut of the given size with uniform sampling and the QAOA distributions of Fig.~\ref{fig:mumbai}(b).
\label{tab:sampled_cuts}}
}
\end{table*}

We also investigate QAOA on the 27 qubit device \emph{ibmq\_mumbai}.
Here, we only use a hardware native graph with random edge weights chosen from $\{-1, 1\}$ since such graphs already require 56 CNOT gates per QAOA layer and \emph{ibmq\_mumbai} has an average CNOT gate error of $0.66\pm  0.14\%$.
A hardware native graph requires six CNOT layers with on average $28/3$ CNOT gates per layer. 
Following Eq.~(\ref{eqn:fidelity_criterion}) the maximum number of layers with $\epsilon=10^{-1}$ is thus 13 when cross-talk effects from applying gates in parallel are neglected.
Here, we set the learning rate and perturbation of SPSA to 0.005 and 0.01, respectively.
Depth-one QAOA (which can be classically simulated efficiently) distinguishes noise from signal and we observe an improved average cut value compared to random sampling.
Depth-two QAOA shows signs of convergence but is impacted by noise which causes large energy jumps.
However, it still produces a lower energy state than depth-one QAOA.
Details are given in Appendix~\ref{sec:graph_details}.

Evaluating the cut value of each of the $2^{27}$ possible solutions is still numerically feasible.
There are 212, 12, and 2 solutions with a cut value of ten, eleven and twelve (the max-cut), respectively.
Therefore, the probability of measuring a cut with a value of ten or more by random sampling is $1.6\cdot10^{-6}$.
Out of the $2^{14}$ cuts sampled from the optimal depth-two state 7 cuts had a maximum value which corresponds to a probability of $0.043\%$, see Tab.~\ref{tab:sampled_cuts}. 
Here, a cut with a value of ten corresponds to an approximation ratio of $83\%$ which is close to, but lower than, the Goemans-Williamson approximation ratio of $\sim 88\%$~\cite{Goemans1995}.
According to the criterion in Eq.~(\ref{eqn:fidelity_criterion}) with $\epsilon=10^{-1}$ depth-two is just within reach.

\section{Discussion and Conclusion\label{sec:discussion}}

We have investigated swap strategies to implement dense circuits built from commuting two-qubit gates on linear, grid, and heavy-hex coupling maps.
For QAOA higher connectivity is not always synonymous with lower gate count due to simplifications between $R_{ZZ}$ and SWAP gates.
Crucially, a line swap strategy is better on a two-dimensional grid than the grid strategy we put forward.
However, our swap strategies on grid coupling maps with dimension three or higher reduce the circuit depth compared to a linear strategy.
Furthermore, the heavy-hex coupling map is almost identical to a linear coupling map.
We note that these strategies may not be optimal and that better strategies may exist, especially for low-density problems.
Crucially, the ability to move logical variables through the physical coupling map means that digital quantum computers do not incur an embedding overhead in the number of qubits as do quantum annealers~\cite{Konz2021}.

The fidelity estimates in Sec.~\ref{sec:fidelity} show that dense optimization problems will almost certainly require gates with an error rate far below fault-tolerance thresholds.
Furthermore, evaluating the depth criterion in Eq.~(\ref{eqn:fidelity_criterion}) with gate fidelities measured in isolation may yield a depth bound that is too optimistic.
This also shows a need for application tailored hardware benchmarks~\cite{Hashim2022, Santra2022}.

In Sec.~\ref{sec:runtime} we provided a methodology to estimate the execution time of QAOA on noisy hardware which we found to be significant.
We caution the reader that these numbers contain a large amount of uncertainty and can be impacted by noise~\cite{Lotshaw2022}. 
This methodology can guide the development of the control hardware.
For example, for the large number of shots ($>10^3$) that QAOA needs the initialization time of the classical control hardware~\cite{Wack2021} is negligible if kept below a second per iteration since the duration of a single shot is likely larger than $1~{\rm ms}$ for problems of practical interest, see Fig.~\ref{fig:tau_qaoa}.
Nevertheless, the duration of the two-qubit gate and the number of shots needed to estimate the energy significantly impact the execution time.
Pulse-efficient gate implementations may help reduce the execution time of QAOA circuits by leveraging the $R_{ZX}$ gate instead of the CNOT gate~\cite{Alexander2020, Stenger2021, Earnest2021}.
The execution time estimation also depends on the variant of QAOA.
For example, counteradiabatic driving reduces $p$ by adding counteradiabatic gates and an extra parameter at each layer~\cite{Chandarana2021, Wurtz2021} while Recursive-QAOA~\cite{Bravyi2020} will also produce different execution time estimates.
However, the extra gates and different number of parameters to optimize may impact the execution time as well.
By contrast, fault-tolerant architectures require a different execution time estimation methodology, as discussed in Ref.~\cite{Sanders2020} which also found that faster error correcting codes are needed to make heuristics for combinatorial optimization competitive.
This also suggests that algorithmic improvements to QAOA~\cite{Chandarana2021, Amaro2022} such as warm-start methods~\cite{Egger2021, Tate2020, Tate2021} will be required to get a quantum advantage in combinatorial optimization with heuristic algorithms.

We demonstrated a Qiskit Runtime program for QAOA on a cloud-based quantum computer.
Using QAOA-tailored transpiler methods we significantly reduced the gate count and duration of the underlying schedules.
This resulted in cut distributions biased towards high-value cuts when compared to random sampling.
Here, we caution that Goemans-Williamson randomized rounding and related procedures are a more meaningful benchmark for large problems~\cite{Goemans1995}.
Ultimately, our results are limited by the fidelity of the cross-resonance gate and decoherence.
CVaR aggregation made it possible to observe convergence at depth-two with 27 qubits.
Deeper QAOA circuits may yet be possible, often at the expense of more shots, using advanced error mitigation methods~\cite{Vazquez2022} such as Pauli-Twirling~\cite{Kim2021}, Probabilistic Error Cancelation~\cite{Berg2022}, M3 readout error mitigation~\cite{Nation2021}, and Zero-Noise Extrapolation~\cite{Temme2017, Kandala2018}.

\section{Acknowledgements}

The authors acknowledge Sergey Bravyi, Giacomo Nannicini, and Libor Caha for useful discussions.
This work was also supported by the Hartree National Centre for Digital Innovation program, funded by the Department for Business, Energy and Industrial Strategy, United Kingdom.
We acknowledge the use of IBM Quantum services for this work.
Code availability: The swap strategies have been implemented in Qiskit as transpiler passes for blocks of commuting two-qubit gates~\footnote{The following pull requests contain the swap strategies as implemented in Qiskit: \href{https://github.com/Qiskit/qiskit-terra/pull/7813}{PR \#7813}, and \href{https://github.com/Qiskit/qiskit-terra/pull/7979}{PR \#7979}.}.
The views expressed are those of the authors, and do not reflect the official policy or position of IBM or the IBM Quantum team.

IBM, the IBM logo, and ibm.com are trademarks of International Business Machines Corp., registered in many jurisdictions worldwide.
Other product and service names might be trademarks of IBM or other companies.
The current list of IBM trademarks is available at https: //www.ibm.com/legal/copytrade.

%\newpage

\bibliographystyle{unsrtnat}
\bibliography{references}

%\newpage

\appendix

\section{SWAP Strategy Details}\label{sec:swap_strategy_details}
Here, we discuss details of the swap strategies summarized in Sec~\ref{sec:swap_strategies} of the main text.
$G_0=(V,E_0)$ is the graph of a coupling map and $N(e)$ the set of edges adjacent to edge $e$.
Then, any swap layer $S_i$ executable on the coupling map is defined by a subset of edges $S_i \subseteq E_0$ that satisfy 
$e \not\in N(e')$ $\forall\, e, e' \in S_i$.
A swap strategy on $G_0$ is then a series of swap layers $\{S_i\}$ where $i\in\mathbb{N}$.
We now discuss swap strategies for the line, grid and heavy-hex coupling maps.

\subsection{Line}

We consider the line graph of size $n$ with vertices numbered from $0$ to $n-1$ and the swap strategy shown in  Fig.~\ref{fig:line_strategy} of the main text.

\begin{lemma}\label{lem:swap_strategy_line}
For the line graph of size $n$ the swap strategy that alternates between two swap layers, one on all odd numbered edges, and one on all even numbered edges reaches full connectivity in $n-2$ layers and is optimal.
\end{lemma}

\begin{proof}
Let $q_i$ denote the $i$-th qubit in the line.
Following Lemma~\ref{lem:swap_strategy_line}, each qubit moves continuously in one direction until reaching a line end where it reverses direction.
Further, starting with the swap layer applied to the even numbered edges, the odd and even numbered qubits begin by moving left and right, respectively. 
The position of qubit $q_i$ starting at node $i$ after applying $k\leq n$ swap layers is
\begin{displaymath}
p_k(q_i) = \begin{cases} \min( i + k, 2n-i-1-k ) \quad &\text{ if } i \text { even }  \\ \max( i - k, -i + k -1 )  &\text{ if } i \text { odd } \end{cases}
\end{displaymath}
for $k \in \{0,\dots,n\}$.
Setting $k = n$ in the equation above, it holds that $p_n(q_i) = n-1-i$ for all $i \in {0, \dots, n-1}$.
Then, after $n$ steps, the line is fully reversed implying that any two qubits were positioned next to each other at some point during the process.
This shows that the strategy reaches full connectivity after $n$ steps.

It turns out, that we already reach full connectivity after $n-2$ steps.
To prove this, note that after $n-2$ steps the position of $q_i$ is
\begin{align}\label{eqn:line_position}
p_{n-2}(q_i) = 
\begin{cases} 
n-2 \quad &\text{ if } i = 0  \\ 
n - i +1   &\text{ if } i \text { even }, i > 0\\
n - i - 3   &\text{ if } i \text { odd }, i < n-2\\
i - n +2   &\text{ if } i \text { odd }, i \geq n-2\\
\end{cases}
\end{align}
First, consider qubit $q_0$, i.e. case $i=0$ in Eq.~(\ref{eqn:line_position}), which starts in the leftmost position.
Its final position is the second rightmost node in the line.
Therefore, during the process $q_0$ passes all nodes but the rightmost one.
It follows that $q_0$ must have been positioned next to every other qubit at some point of the swap process.
The qubit initially at $i\geq n-2$ with odd $i$, i.e. the fourth case in Eq.~(\ref{eqn:line_position}), arrives in position 0 or 1 after $n-2$ SWAP layers, i.e. $p_{n-2}(q_i) \in \{0,1\}$.
It was therefore positioned next to every other qubit at some point of the SWAP strategy.
Now consider the second and third cases in Eq.~(\ref{eqn:line_position}).
When $i_1, i_2 \in {1, \dots, n-1}$, such that qubit $q_{i_1}$ is to the left of qubit $q_{i_2}$, i.e. $i_1 < i_2$, and neither $i_1$ nor $i_2$ are larger than or equal to $n-2$ and odd, then, after $n-2$ steps of the swap strategy $p_{n-2}(q_{i_1}) - p_{n-2}(q_{i_2})$ is equal to
\begin{align}\label{eqn:line_cases}
\begin{cases} 
i_2 - i_1 \quad &\text{ if } i_1 = i_2 \mod 2  \\ 
i_2 - i_1 + 4   &\text{ if } i_1 \text { even and } i_2  \text { odd }\\
i_2 - i_1 - 4   &\text{ if } i_2 \text { even and } i_1  \text { odd }\\
\end{cases}
\end{align}
In the first two cases $p_{n-2}(q_{i_1}) - p_{n-2}(q_{i_2}) > 0$, i.e. $q_{i_1}$ is now to the right of $q_{i_2}$ in the line, since by assumption $i_1 < i_2$. 
Hence, the corresponding qubits have switched their order within the line and must have been adjacent to each other at some point during the strategy.
If $i_2$ is even and $i_1$ is odd, it suffices to consider the case where $i_1 < i_2 -1$, since otherwise the corresponding qubits are initially adjacent.
Then $i_2 - i_1 \geq 3$ and Eq.~(\ref{eqn:line_cases}) implies $p_{n-2}(q_{i_1}) - p_{n-2}(q_{i_2}) \geq -1$.
Hence, the corresponding qubits either end up adjacent to one another or have switched their order after $n-2$ steps. This proves that we reach full connectivity after $n-2$ swap layers.

To prove optimality, consider qubit $q_0$ starting in the leftmost position of the line.
At any point of the process, all qubits left of $q_0$ have been adjacent to $q_0$ at some previous point.
In particular, this holds for the adjacent qubit to the left of $q_0$.
By this argument and as every node in the line graph has at most degree two, $q_0$ can only become connected to at most one additional qubit after every additional swap layer, namely by the edge to its right.
Since $q_0$ is initially only connected to $q_1$ and needs to be connected to $n-2$ additional qubits to reach full connectivity, it then follows that no swap strategy with less that $n-2$ layers can lead to full connectivity in the line graph.
\end{proof}

\subsection{Grid}

The swap strategy for the two-dimensional grid is an extension of the line swap strategy.

\begin{lemma}\label{lem:swap_strategy_grid}
For the two-dimensional grid of size $n$, there exists a swap strategy that reaches full connectivity in $n/2 + \mathcal{O}(\sqrt{n})$ layers.
For the three dimensional grid of size $n$ a strategy with depth $n/4 + \mathcal{O}(n^{2/3})$ exists.
\end{lemma}

\begin{proof}
First, consider two equally long adjacent horizontal lines of qubits, where each qubit in the lower line is connected to the corresponding qubit in the upper line.
We apply the line strategy in Lemma~\ref{lem:swap_strategy_line} to both lines, where we begin by applying SWAP gates to even numbered edges in one line and odd numbered edges in the other, see Fig~\ref{fig:double_line_strategy}.
Since both lines are reversed after $n$ steps, any two nodes in the upper and lower line were adjacent at some point. 
Crucially, this is only possible because SWAP gates for odd edges on one line are executed simultaneously with SWAP gates on even edges of the other line and vice versa.
Thus, full connectivity is reached after at most $n$ steps.
Additionally, since both lines fully reverse no new connections are obtained in the last step and full connectivity is already reached after $n-1$ steps.

Second, consider the square grid with $n = x^2$ nodes divided into $x$ rows and columns.
The grid swap strategy, shown in Fig.~\ref{fig:grid_strategy} of the main text, repeats two steps.
\begin{enumerate}
    \item Apply $x-1$ steps of the line swap strategy separately to each row.
    Importantly, on two adjacent rows the SWAP gates must never simultaneously be on edges with the same parity.
    \item Swap the rows by applying exactly two steps of the line swap strategy to each column in parallel.  
\end{enumerate}
The double line example shows that after executing the first step a qubit in a row is connected to all the other qubits in its row and the neighboring rows.
The second step of the strategy is executed with two swap layers.
It swaps rows such that every row is now positioned next to two different rows.
Thus, step 1 connects qubits of adjacent rows and step 2 shuffles the order of rows such that all rows are adjacent at some point.
Since we perform two vertical swap layers in each iteration we reach full connectivity after repeating both steps $(x-2)/2$ times and step 1 one additional time at the end.
We thus need
\begin{displaymath}
\frac{x-2}{2} (x+1) + (x-1) =\frac{n}{2} + \mathcal{O}(\sqrt{n})
\end{displaymath}
swap layers, which proves the two-dimensional case.
The proof for the three-dimensional case is similar.
\end{proof}

\begin{figure}
\centering
\includegraphics[width=\columnwidth]{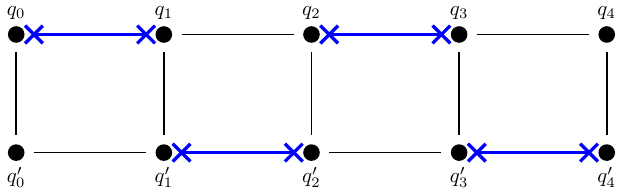}
\caption{
Swap strategy in a graph with two connected lines each with five vertices.
The blue SWAP gates show the first swap layer.
The qubit order in each line fully reverses after five steps and full connectivity is reached after four steps.
\label{fig:double_line_strategy}}
\end{figure}

\subsection{Heavy-hex}

A heavy-hex coupling map has a mixture of degree-two and degree-three nodes.
Its qubits are placed on the edges and vertices of hexagons~\cite{Chamberland2020}.
Each hexagon therefore has 12 qubits.
Here, we focus on $i\times j$ heavy-hex graphs which have $i$ rows and $j$ columns of hexagons, as exemplified by the $3\times 3$ heavy-hex graph in Fig.~\ref{fig:heavy_hex_longest_line}.
The total number of qubits is $n=5ij+4(i+j)-1$ and the length of the longest line in the graph is $l_\text{max}=4(ij+i+j)+1$.
The length of this line is related to the total number of qubits by $l_\text{max}=4[n(i,j)+i+j + 1]/5 + 1$ which, to leading order, scales as $4n/5$.
Furthermore, $l_\text{max}$ is bound from above by
\begin{align}\label{eqn:hh_line_bound}
\frac{4}{5}n+\frac{4}{5}\sqrt{n}+1
\end{align}
when the grid of hexagons is approximately square, i.e. $i\sim j$.
These preliminaries allow us to formulate the following Lemma for the heavy-hex swap strategy.

\begin{figure}
\centering
\includegraphics[width=\columnwidth]{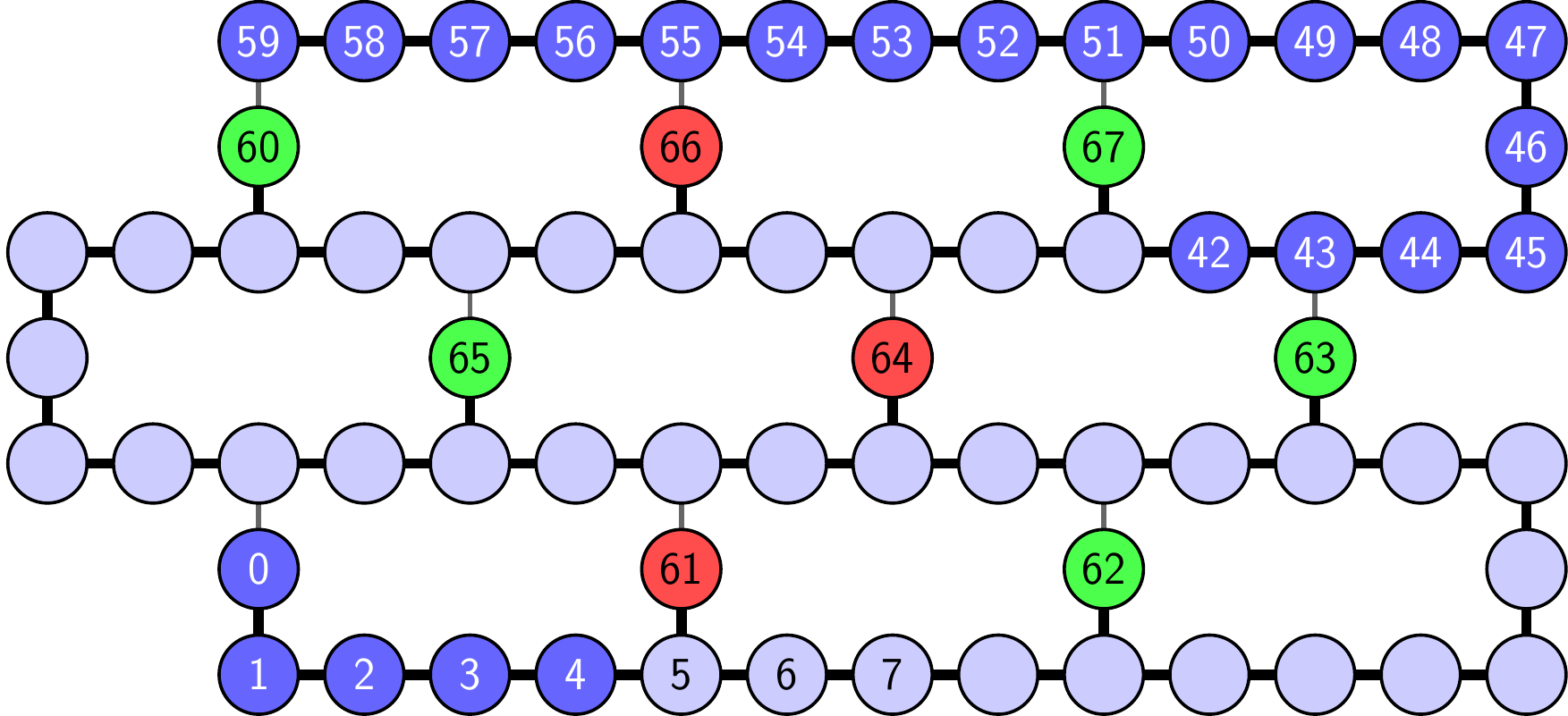}
\caption{
A three by three heavy-hex coupling map.
The longest line ranges from qubits 0 to 60.
By contrast, the blue qubits are part of the longest-line of the \emph{unfolded} heavy-hex graph with $l=60$.
The dark blue nodes indicate the tails with five and $t=18$ qubits.
The red and green nodes indicate the $A$ and $B$ qubits, respectively.
The grey edges are the edges that are removed to unfold the heavy-hex graph.
}\label{fig:heavy_hex_longest_line}
\end{figure}

\begin{lemma}\label{lem:heavy_hex}
For the approximately square heavy-hex grid of size $n$, there exists a swap strategy that reaches full connectivity in less than $n+\sqrt{n}+61$ swap layers.
\end{lemma}
\begin{proof}
We prove Lemma~\ref{lem:heavy_hex} by unfolding the heavy-hex graph along a line of length $l$ where $l \mod 4 = 0$, see Fig.~\ref{fig:heavy_hex_longest_line}.
To unfold, we delete one edge connected to the nodes not in the longest line. 
The result is a line graph with an optional additional node connected to every fourth node, see Fig.~\ref{fig:heavy_hex_vertex_groups}.
The graph additionally has a tail of $5$ vertices on one side and a tail of $t$ vertices on the other side, such that $t = 2 \text{ mod } 4$, see Fig.~\ref{fig:heavy_hex_longest_line}.
Here, $t$ depends on the width of the heavy-hex graph. 
We will prove that Lemma~\ref{lem:heavy_hex} holds for any graph of this kind.

The proof applies a line swap strategy on the unfolded heavy-hex graph modified such that at any time some qubits remain in the positions of the additional nodes without moving along the line.
We divide the process into five iterations.
In each iteration a qubit either moves along $1/4$ of the line or waits in one of the adjacent nodes.
If every qubit only enters a waiting state once during the complete process, all qubits will have completed $l$ steps of the simple line strategy after five iterations, leading to full connectivity.
The difficulty then lies in ensuring that every qubit is swapped into a waiting position at most once.
We divide the additional nodes into two sets $A$ and $B$ spaced apart by eight nodes in the line, see Fig.~\ref{fig:heavy_hex_vertex_groups}.
We number the edges in the line by their position and define four swap layers to reach full connectivity.
\begin{itemize}
    \item $S_1$: All odd-numbered edges in the line.
    \item $S_2$: All even-numbered edges in the line.
    \item $S_3$: All edges connected to vertices in group $A$.
    \item $S_4$: All edges connected to vertices in group $B$.
\end{itemize}
We claim that the swap strategy in which we
\begin{enumerate}
    \item alternate between $S_1$ and $S_2$ $k-7$ times starting with $S_1$, and
    \item apply $S_4$ once, and
    \item alternate between $S_1$ and $S_2$ $7$ times starting with $S_2$, and
    \item apply $S_3$ once, and
    \item repeat Steps 1-4 five times,
\end{enumerate}
reaches full connectivity on the unfolded heavy-hex if
\begin{align}\label{eqn:hh_k}
k = \frac{l}{4} - \left( \frac{l}{4} \text{ mod } 8 \right) + 10.
\end{align}
Here, $k$ is chosen such that $k~{\rm mod}~8=2$ and $k>l/4$ so that every qubit travels the full line with four iterations of steps 1. and 3. Steps 2. and 4. swap qubits in and out of nodes $A$ and $B$, respectively.
Steps 1. to 4. require $k+2$ swap layers so that after five iterations the total number of swap layers is $5k+10$ which is thus bounded from above by $n+\sqrt{n}+61$ as seen by injecting Eq.~(\ref{eqn:hh_line_bound}) in Eq.~(\ref{eqn:hh_k}) and conservatively assuming $l/4~{\rm mod}~8=0$.

As argued above, it suffices to show that any particular qubit will remain in a position corresponding to groups $A$ or $B$ for at most one iteration.
We assign each node along the line to one of eight evenly-spaced groups $V_i$ with $i \in \{0, \dots , 7\}$.
The set of all vertices is thus partitioned into ten sets, $A$, $B$ and $V_i$, see Fig.~\ref{fig:heavy_hex_vertex_groups}.
We now examine how the qubits switch groups during steps 1. to 4.
Note that during each iteration $k$ steps of the line strategy are executed and $k$ was chosen, such that $k \text{ mod } 8 =2$.

We first ignore groups $A$ and $B$. 
Since $k$ is even and we begin by applying SWAP gates to odd numbered edges in the line, i.e. $S_1$,
qubits moving towards the left will switch from group $V_i$ to $V_{i-2}$ and qubits moving towards the right will move from $V_i$ to $V_{i+2}$.
If a qubit reaches the end of the line during an iteration it will either
\begin{itemize}
    \item switch direction from left to right and switch groups from $V_i$ to $V_{7-i}$, or
    \item switch direction from right to left and switch groups from $V_i$ to $V_{7-i+4}$.
\end{itemize}
This behavior results from $k~{\rm mod}~8=2$.
The movement pattern thus undergone by individual qubits is depicted in Fig.~\ref{fig:qubit_movement_pattern}.

\begin{figure}[htbp!]
     \centering
    \includegraphics[width=\columnwidth]{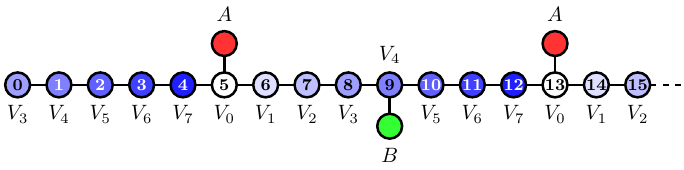}
    \caption{
    The unrolled heavy-hex graph with numbered nodes.
    The nodes are divided into ten subsets $A$ (in red) $B$ (in green) and $V_i$ for $i \in \{0,\dots,7\}$ (from white to blue).
    \label{fig:heavy_hex_vertex_groups}
    }
\end{figure}

\begin{figure}[htbp!]
    \centering
    \includegraphics[width=0.8\columnwidth,clip,trim=0 10 0 0 ]{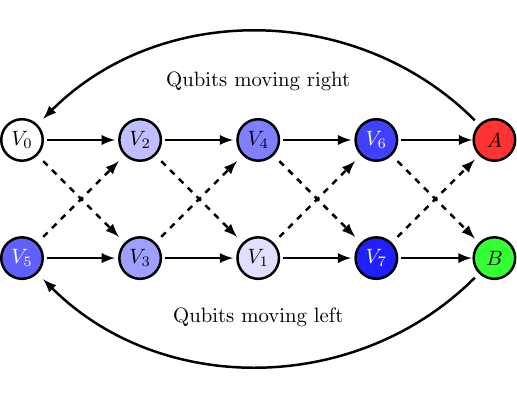}
    \caption{
    Qubit movement pattern during the heavy-hex swap strategy.
    Each node shows one of the ten categories of nodes defined in Fig.~\ref{fig:heavy_hex_vertex_groups}.
    The arrows show how the qubits change group after an iteration of steps 1. through 4. 
    Solid and dashed lines show qubits that were and were not reflected at line ends during the iteration, respectively.
    Note that a dashed arrow from $V_6$ to $V_5$ is not shown.
}\label{fig:qubit_movement_pattern}
\end{figure}

We now consider groups $A$ and $B$.
During one iteration qubits positioned in groups $A$ and $B$ will switch to groups $V_0$ and $V_5$, respectively.
From the outlined strategy, qubits only arrive in group $A$ if they ended the previous iteration in a group $V_i$ with 
\begin{align*}
i + 2 &= 0 \text{ mod } 8 &&\implies i = 6 \\
\text{or} \quad  7 - i &=  0 \text{ mod } 8  &&\implies i = 7.
\end{align*}
Since the swap layer swapping with group $B$ is executed after $k-7$ layers of $S_1$ in alternation with $S_2$ and nodes of group $B$ are connected with group $V_4$, a qubit can only arrive in group $B$ if it was initially positioned in a group $V_i$ with
\begin{align*}
i - 2 + 7 &= 4 \text{ mod } 8 &&\implies i = 7\\
\text{or} \quad  7 - i + 4 + 7 &=  4 \text{ mod } 8  &&\implies i = 6.
\end{align*}
In either case, a qubit can only arrive in groups $A$ or $B$ if it was previously positioned in group $V_6$ or $V_7$.
All possible movements of qubits among groups from one iteration to the next are thus captured by Fig.~\ref{fig:qubit_movement_pattern} and it is clear that within five iterations each qubit will only visit $A$ or $B$ at most once.
This concludes the proof.
\end{proof}

\section{Circuit depth and CNOT gate count\label{sec:circuit_depth_and_gate_count_details}}

We now investigate how the number of CNOT layers and gates in QAOA circuits, transpiled using the pass described in Sec.~\ref{sec:hardware_optimized_transpiler_pass}, scale with problem size. 
Here, $G_0=(V,E_0)$ is either a heavy-hex or an $\eta$-dimensional grid coupling map and $\{S_i\}$ with $i \in \mathbb{N}$ is a swap strategy compatible with $G_0$.
The transpiled circuit has alternating layers of $R_{ZZ}$ and SWAP gates applied on edge sets $E_i$ and $S_i$, respectively, see Fig.~\ref{fig:transpiled_qaoa_cost_layer}.
We determine the number of CNOT gates and layers by counting gates and layers in $E_i$ and $S_i$ taking into account gate cancellations across layers.

In every layer $E_i$ we apply $R_{ZZ}$ gates on the edges in $E_i$.
The number of edges in $E_i$ thus determines the number of $R_{ZZ}$ gates.
Here, $E_i$ is the set of edges in the hardware coupling map (i.e. $E_i\subseteq E_0$) which give new qubit connections after applying the swap layer $S_{i-1}$.
Edges which may give new qubit connections after $S_{i-i}$ are in the neighbourhood of swapped edges, i.e.
\begin{align}\label{eqn:swap_layer_subset}
E_i \subseteq \bigcup_{e \in S_{i-1}} N(e),
\end{align}
as exemplified in Fig.~\ref{fig:notation_example}.
Since the SWAP gates in a swap layer $S_i$ are executed in parallel $S_i$ never contains two neighboring edges, i.e. $e \notin N(e')~\forall~e, e' \in S_{i}$.
Equation~(\ref{eqn:swap_layer_subset}) therefore implies $E_i \subseteq E_0 \setminus S_{i-1}$.

\begin{figure}[htbp!]
\centering
\includegraphics[width=0.95\columnwidth]{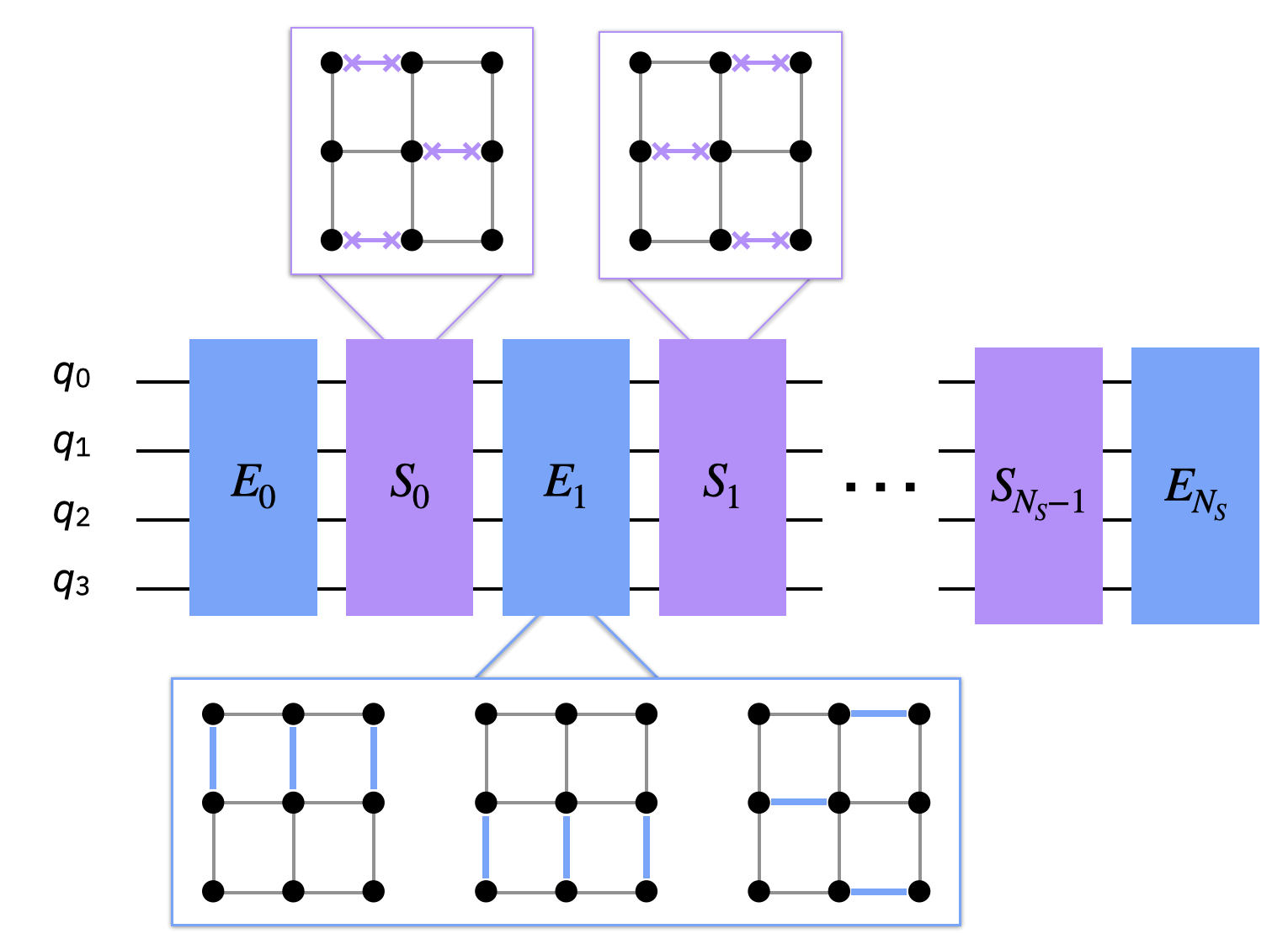}
\caption{
A QAOA cost layer transpiled with the hardware-optimized transpiler pass from Sec.~\ref{sec:hardware_optimized_transpiler_pass}. 
The circuit is split into blocks $E_i$ and layers $S_i$ of $R_{ZZ}$ and SWAP gates, respectively, exemplified on a two-dimensional grid.
Here, the last $R_{ZZ}$ layer of $E_i$ can be combined with $S_i$ by the circuit identity in Fig.~\ref{fig:swap_rzz_gate_cancellation}. \label{fig:transpiled_qaoa_cost_layer}
} 
\end{figure}

Carefully positioned CNOT gates across $E_i$ and $S_i$ may cancel.
We therefore position all $R_{ZZ}$ gates of $E_i$ that are applied across edges contained in $S_i$ at the end of $E_i$.
This allows us to simplify two CNOT gates and implement $\rm{SWAP}\cdot R_{ZZ}(\theta)$ with three CNOTs and a $R_Z(\theta)$ gate, see Fig.~\ref{fig:swap_rzz_gate_cancellation}.
The number of CNOT layers $L_{cx}(E_i, S_i)$ required to implement $E_i$ and $S_i$ is thus bounded by the edge chromatic number $\chi'(G_i)$ of the graph $G_i = (V, E_i \setminus S_i)$, i.e.
\begin{align}\label{eqn:edge_chromatic_number_bound}
L_{cx}(E_i, S_i) &\leq 2\chi'(G_i) + 3.
\end{align}
Here, the $R_{ZZ}$ gates in $E_i\setminus S_i$ are divided into $\chi'(G_i)$ sets of gates to execute in parallel each with two CNOT gates.
The extra three CNOT layers come from SWAP gates that may have absorbed $R_{ZZ}$ gates.
Similarly, the number of CNOT gates in $E_i$ and $S_i$ is 
\begin{align}\label{eq:gate_count_bound}
k_i = 2\lvert E_i \setminus S_i \rvert + 3 \lvert S_i \rvert
\end{align}
since a swap layer will always require three CNOT layers even after absorbing $R_{ZZ}$ gates from $E_i$.
Since $E_i \subseteq E_0 \setminus S_{i-1}$ we may write $|E_i\setminus S_i|\leq|(E_0\setminus S_{i-1})\setminus S_i|=|E_0|-|S_{i-1}|-|S_i|$ where the last equality requires that $S_{i-1}$ and $S_i$ are disjoint.
Under the same assumptions we also have $\chi'(G_i)\leq\chi'(G)-2$.
Therefore, Eq.~(\ref{eqn:edge_chromatic_number_bound}) and~(\ref{eq:gate_count_bound}) yield
\begin{lemma}\label{lem:cnot_depth_gate_count_bound}
If $S_i\cap S_{i-1}=\varnothing$ the number of CNOT gates $k_i$ required to implement $E_i$ and $S_i$ satisfies
\begin{displaymath}
k_i \leq 2\lvert E_0 \rvert - 2\lvert S_{i-1} \rvert  + \lvert S_i \rvert.
\end{displaymath}
If there exists an edge coloring of the coupling map $E_0$ with $\Delta_i$ colors, such that one color corresponds to the set of edges contained in  $S_{i-1}$ and all edges of $S_i$ are colored in the same color, then the number of CNOT layers $L_\text{cx}(E_i, S_i)$ required to implement $E_i$ and $S_i$ satisfies
\begin{align*}
L_\text{cx}(E_i, S_i) &\leq 2\Delta_i-1.
\end{align*}
\end{lemma}

\begin{figure}[t]
    \centering
    \includegraphics[width=\columnwidth]{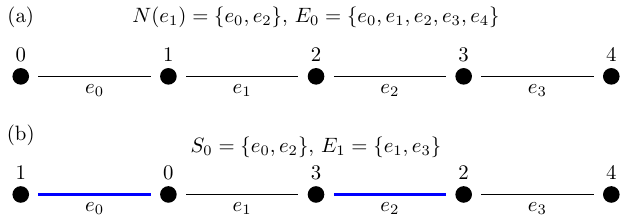}
    \caption{Explanation of the notation on a line graph. (a)~A line graph with five nodes and, as example, $N(e_1)$. (b) Notation after the first SWAP layer. The set of edges that give new qubit connections is $E_1=\{e_1, e_3\}$.}
    \label{fig:notation_example}
\end{figure}

\begin{figure}[tbp!]
    \centering
    \includegraphics[width=0.9\columnwidth]{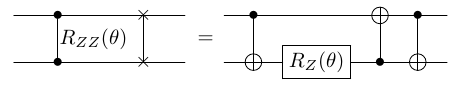}
    \caption{An $R_{ZZ}$ and a SWAP gate combine into a gate sequence with three CNOTs and one $R_Z$ gate. \label{fig:swap_rzz_gate_cancellation}} 
\end{figure}

\subsubsection{Grid coupling maps}

We now apply Lemma~\ref{lem:cnot_depth_gate_count_bound} to the swap strategies on the line, two- and three-dimensional grids with $L_S$ swap layers.
These strategies satisfy $S_i\cap S_{i-1}=\varnothing\,\forall\,i$ and have $S_{-1} = S_{L_S} = \varnothing$, see Appendix~\ref{sec:swap_strategy_details}.
The number of CNOT layers in the complete cost layer is then bounded by
\begin{align} \notag
k = \sum_{i = 0}^{L_{S}} k_i
&\leq \sum_{i = 0}^{L_{S}}  2\lvert E_0 \rvert - 2\lvert S_{i-1} \rvert + \lvert S_i \rvert\\
&= 2 \lvert E_0 \rvert (L_S  +1) - \sum_{i=0}^{L_S} \lvert S_i \rvert \label{eqn:k_computed}
\end{align}
In the line graph, every SWAP layer contains half of the edges (on average in case the number of edges is odd).
More generally, for the $\eta$-dimensional grid strategy the number of swap gates is, on average,
\begin{align}\label{eqn:graph_averages}
\lvert S_i \rvert = \frac{1}{2\eta}\lvert E_0\rvert ~\text{where}~ \lvert E_0 \rvert = \eta n + \mathcal{O}(n^{\frac{\eta -1}{\eta}}).
\end{align}
Combining Eq.~(\ref{eqn:k_computed}) and (\ref{eqn:graph_averages}) yields the bound 
\begin{align}
    k\leq&\, \frac{4\eta - 1}{2}n(L_S+1)+\mathcal{O}(L_Sn^{\frac{\eta -1}{\eta}})
\end{align}
on the total number of CNOT gates $k$ required to implement $\exp(-i\gamma H_C)$.
These bounds are summarized in Tab.~\ref{tab:depth_gate_count} of the main text as $(4\eta-1)/2nL_S$.
Furthermore, since the swap layers in the grid strategy for the $\eta$-dimensional grid graph exactly correspond to an edge coloring with $2\eta$ colors, the second requirement from Lemma~\ref{lem:cnot_depth_gate_count_bound}, i.e. the existence of the edge coloring, also holds with
\begin{displaymath}
\Delta_i = \begin{cases}
2 \eta  &0 < i < L_S\\
2 \eta + 1 &i = 0.
\end{cases}
\end{displaymath}
Together with the $2(2 \eta -1)$ CNOT layers required to implement the final layer $E_{L_S}$, this bounds the total number of CNOT layers $L_\text{cx}$ of a cost layer with $L_S$ swap layers following
\begin{align*}
L_\text{cx} &\leq 2(2 \eta -1)  + \sum_{i=0}^{L_S-1}2 \Delta_i -1 \\ 
&= (4 \eta -1)(L_S+1) + 1.
\end{align*}
Plugging in $\eta = 1,2,3$ gives the numbers in Tab.~\ref{tab:depth_gate_count} of the main text.

\subsubsection{Heavy-hex coupling maps}

We now consider heavy-hex graphs with the same number of rows and columns and with $n$ qubits.
This graph has $\lvert E_0 \rvert = 6n/5 + \mathcal{O}\left(\sqrt{n}\right)$
edges and a longest line with length
$l_{max} = 4n/5 + \mathcal{O}\left(\sqrt{n}\right)$,
as exemplified in  Fig.~\ref{fig:heavy_hex_longest_line}.
Thus, each SWAP layer $S_i$ consists of at most $2n/5 + \mathcal{O}\left(\sqrt{n}\right)$ SWAP gates. 
It also holds that 
\begin{displaymath}
\lvert E_i \setminus S_i \rvert + \lvert S_i \rvert \leq \frac{3}{5}n + \mathcal{O}\left(\sqrt{n}\right)
\end{displaymath}
for all $0 <i < L_S$.
Using Eq.~(\ref{eq:gate_count_bound}), we can therefore bound the number of CNOT gates $k_i$ in the combined layers $E_i$ and $S_i$ for $0 < i < L_S$ by
\begin{displaymath}
k_i \leq \frac{8}{5}n + \mathcal{O}\left(\sqrt{n}\right).
\end{displaymath}
$E_0$ and $S_0$ contain at most $14n/5 + \mathcal{O}(\sqrt{n})$ and $E_{L_S}$ at most $12n/5 + \mathcal{O}(\sqrt{n})$ CNOT gates, yielding a bound for the total number of CNOT gates $k$
\begin{displaymath}
k \leq (4L_S+9)\frac{2}{5}n + \mathcal{O}(\sqrt{n}).
\end{displaymath}
Furthermore, every layer $E_i$ can be executed with $2$ individual $R_{ZZ}$ layers.
However, it is not possible to cancel out every gate in any of the resulting $R_{ZZ}$ layers with the subsequent swap layer $S_{i}$.
As a result, we get the following bounds for the total depth and number of CNOT layers $d$ and $d_{CNOT}$ in the heavy-hex case
\begin{align*}
d &\leq 9L_S + 10\\
d_{CNOT} &\leq 7L_S + 6
\end{align*}

\begin{figure}[b!]
    \centering
    \includegraphics[width=0.95\columnwidth]{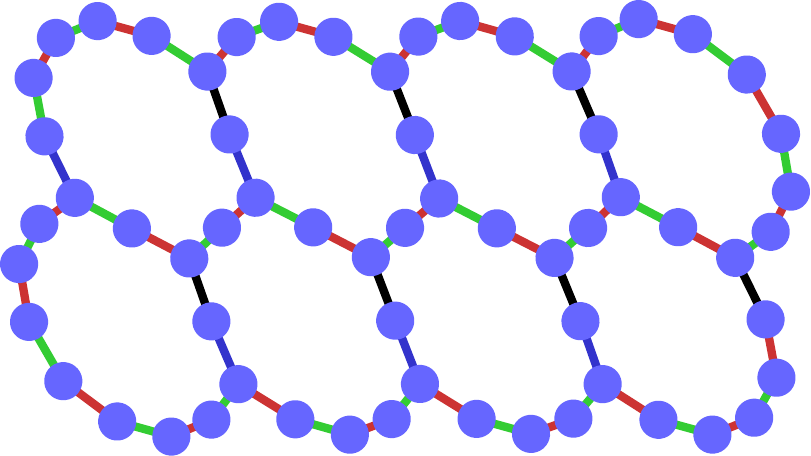}
    \caption{
    Heavy-hex graph with two rows and four columns.
    The color of the edges shows the edge coloring used to prioritize the $R_{ZZ}$ gates which can be simultaneously applied.
    The order of the applied gates is red, green, blue, and black.
    Furthermore, red and green edges form the longest line.
    }
    \label{fig:hh_example}
\end{figure}

\begin{figure*}[htbp!]
    \centering
    \includegraphics[width=\textwidth]{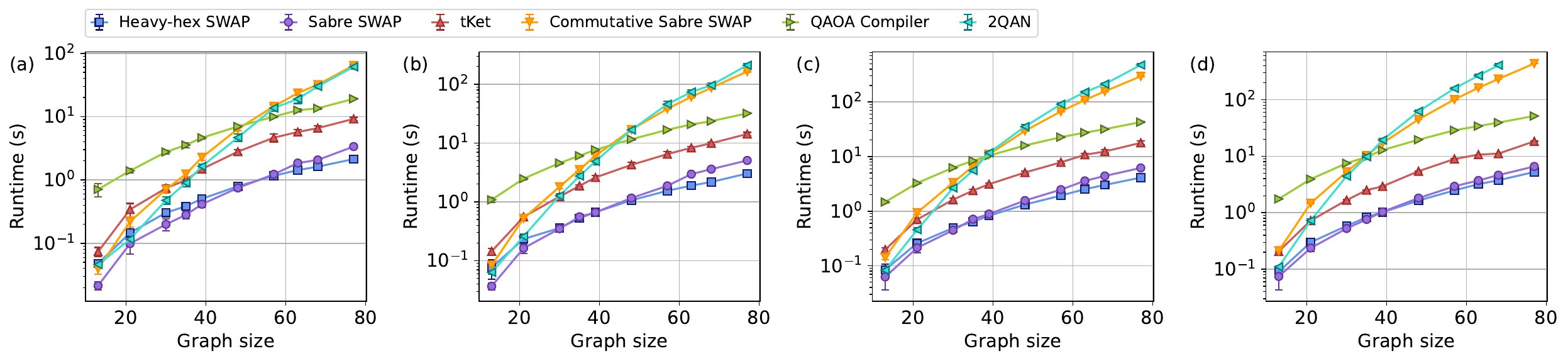}
    \caption{Logarithmic transpiler runtime for QAOA circuits of graph instances with different size and density after transpilation to a heavy-hex coupling map using  SabreSwap, commutative aware SabreSwap, $t|\text{Ket}\rangle$, 2QAN, and the heavy-hex swap strategy.
    Each transpiler was allowed a maximum runtime of ten minutes.
    Each data point is an average over ten random graphs. The lines show the average and the error bars show the standard deviation.
    (a), (b), (c), and (d) correspond to graphs with densities of 0.25, 0.5, 0.75 and 1.0, respectively.
    }
    \label{fig:transpiler_runtime}
\end{figure*}

\section{Commutation aware SabreSwap\label{sec:commute_sabre}}

Quantum circuits can be represented as directed acyclic graphs (DAG) in which nodes are instructions and edges are qubits.
The SabreSwap algorithm traverses the DAG. 
It first builds up a front layer with the gates in the DAG that are first executed on the qubits.
Next, the algorithm inserts SWAPs attempting to minimize the distance between qubits that interact in the front layer.
When two qubits become adjacent the gate from the front layer is inserted and the front layer is updated with the next gates from the DAG.

In the commutation aware version of SabreSwap, the front layer is adapted to contain all upcoming commuting gates obtained from the dependency DAG.
The DAG dependency is built by taking commutation relations into account~\cite{Iten2020}.
Here, two gate nodes are only connected by an edge if the corresponding gates do not commute.

\section{Maximum cut\label{sec:maxcut}}

In the weighted maximum cut problem we are tasked to partition the set of nodes $V$ of a given graph $G=(V,E)$ in two such that the sum of the edge weights $\omega_{i,j}$ with $(i,j)\in E$  traversed by the cut is maximum~\cite{Goemans1995}.
This can be formulated as the QUBO problem
\begin{align}
    \max\frac{1}{2}\sum_{(i,j)\in E}\omega_{i,j}(1-z_iz_j), \notag \\
    \text{such that}\quad z\in\{-1,1\}^{|V|}.\notag
\end{align}
The binary variable $z_i$ indicates which side of the cut node $i$ is.
In Sec.~\ref{sec:runtime_program} we consider weighted maximum cut instances in which the weights can be negative as well as positive~\cite{Hong2008}.

\section{Swap strategy benchmark\label{sec:swap_strat_benchmark_appendix}}

This section provides additional details on the swap strategy benchmark presented in the main text.
For SabreSwap, commutative aware SabreSwap, the QAOA Compiler, 2QAN, and t$|\text{Ket}\rangle$ we generate the \texttt{qaoa\_circuit} using $R_{ZZ}$ gates based on the problem graph.
For the swap strategies we create the \texttt{qaoa\_circuit} as a single instruction built from the \texttt{PauliSumOp} in Qiskit which the transpiler must identify as being made of commuting two-qubit Pauli terms.
In all cases the transpiler must map the instructions to hardware native CNOT gates.
The CNOT circuit depth shown in the main text is computed using the Qiskit function
\begin{minted}[fontsize=\small]{python}
circuit.depth(
    lambda x: isinstance(x.operation, CXGate)
)
\end{minted}
We now list the key passes with which we benchmark the different transpilers.
We obtain the execution time of each transpiler with \texttt{time} as exemplified for t$|\text{Ket}\rangle$.
The t$|\text{Ket}\rangle$ transpiler steps are listed below.
\begin{minted}[fontsize=\small]{python}
tk_circ = qiskit_to_tk(qaoa_circuit)

cmap = Architecture(coupling_map.get_edges())

start = time.time()
placement = GraphPlacement(cmap)
CXMappingPass(cmap, placement).apply(tk_circ)
KAKDecomposition(allow_swaps=False).apply(tk_circ)
final = time.time()

final_circuit = tk_to_qiskit(tk_circ)
\end{minted}

The QAOA Compiler, described in Ref.~\cite{Alam2020}, performs an initial mapping of QAOA and routs the resulting circuit to the hardware using an external compiler. 
The external compiler is called on partial circuits and the results are then stitched together.
We run the QAOA Compiler with Qiskit on optimization level 3 as external compiler and use the Variation-aware Incremental Compilation (VIC) setting described in Ref.~\cite{Alam2020}.
We therefore create an instance of \texttt{CompileQAOAQiskit} and call its \texttt{run\_incr\_c} method with \texttt{variation\_aware=True} and \texttt{ initial\_layout\_method="qaim"}.
The configuration used to initialize the \texttt{CompileQAOAQiskit} instance is
\begin{minted}[fontsize=\small]{python}
config = {
    "Backend": "qiskit",
    "Target_p": "1",
    "Packing_Limit": "10e10",
    "Route_Method": "sabre",
    "Trans_Seed": "0",
    "Opt_Level": "3",
}
\end{minted}
To benchmark the swap strategies in Qiskit we use the pass manager
\begin{minted}[fontsize=\small]{python}
pm = PassManager([
    FindCommutingPauliEvolutions(),
    Commuting2qGateRouter(swap_strat, edge_color),
    Decompose(),
    Decompose(),
    CXCancellation(),
])
\end{minted}
Here, the swap strategy contains the information on the coupling map and the double \texttt{Decompose} replaces the routed two-qubit Pauli evolution operators by CNOT gates.
The \texttt{edge\_color}, exemplified in Fig.~\ref{fig:hh_example}, is used to prioritize which $R_{ZZ}$ gates should be applied first at each step of the swap strategy.
This increases the number of $R_{ZZ}$ gates that are followed-by SWAP gates such that the \texttt{CXCancellation} pass can cancel consecutive CNOT gates as shown in Fig.~\ref{fig:swap_rzz_gate_cancellation}.
For SabreSwap we replace the \texttt{FindCommutingPauliEvolutions} and \texttt{Commuting2qGateRouter} with the single pass
\begin{minted}[fontsize=\small]{python}
SabreSwap(
    coupling_map=coupling_map, 
    heuristic='basic'
)
\end{minted}

The 2QAN~\cite{Lao2021} benchmark, shown below, uses the \texttt{QuRouter} imported from the \texttt{py2qan} package with a trivial initial mapping.
\begin{minted}[fontsize=\small]{python}
router = QuRouter(
    qasm=qaoa_circuit.qasm(),
    init_map={idx: idx for idx in range(n)}, 
    coupling_map=coupling_map.get_edges(),
)

circ_2qan, _ = router.run_qaoa(
    layers=1, 
    gammas=[0.5], 
    betas=[0.75], 
    msmt=True
)
\end{minted}

In addition to the circuit depth and CNOT gate count presented in Fig.~\ref{fig:transpiler_comparison} of the main text we also compute the time it takes to transpile the corresponding circuits.
We find that the commutative aware version of SabreSwap and 2QAN have the runtime with the worst scaling while the swap strategies and SabreSwap are up to two orders of magnitude faster, see Fig.~\ref{fig:transpiler_runtime}.

\section{Experiment details\label{sec:graph_details}}

The seven node graph used in the Qiskit Runtime example can be embedded with one swap layer $S_0=\{(0, 1), (3, 5)\}$ on the seven qubit \emph{ibm\_nairobi} system.
The weights of the graph were chosen at random from $\{-1, 1\}$.
The seven node graph given as sums of Pauli-$Z$ operators is $\mathcal{G}_{10}=-Z_1Z_0+Z_2Z_1-Z_4Z_3-Z_5Z_4-Z_2Z_0+Z_6Z_5+Z_5Z_0+Z_6Z_3-Z_3Z_1+Z_5Z_3$ and has a single maximum cut with a value of three, see Fig.~\ref{fig:nairobi_map}.
The properties of the hardware and the gates as reported by \emph{ibm\_nairobi} on the date the data were acquired are given in Tab.~\ref{tab:nairobi_prop} and~\ref{tab:nairobi_gate}
.

\begin{figure}[htbp!]
    \centering
    \includegraphics[width=0.8\columnwidth]{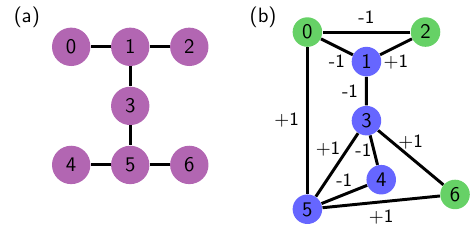}
    \caption{
    (a) Coupling map of \emph{ibm\_nairobi} and (b) the graph $\mathcal{G}_{10}$.
    The node coloring shows the maximum cut.
    }
    \label{fig:nairobi_map}
\end{figure}

\begin{table}[htbp!]
    \centering{
    \begin{tabular}{l r r r r} \hline\hline
 & $T_1$ & $T_2$ & $\sqrt{X}$ Error  & RO Error \\
Qubit & ($\mu{\rm s}$) & ($\mu{\rm s}$) & $\times 10^{-4}$ & \% \\ \hline
0 & 124 &  61 & 2.5 & 2.1 \\
1 & 100 & 132 & 6.2 & 3.2 \\
2 & 120 & 158 & 1.7 & 2.1 \\
3 &  91 &  65 & 4.0 & 3.8 \\
4 &  67 &  83 & 2.9 & 2.4 \\
5 & 132 &  44 & 3.3 & 2.9 \\
6 & 165 & 116 & 2.2 & 2.8 \\ \hline\hline
    \end{tabular}
    \caption{Hardware properties of ibm\_nairobi on 22.12.2021 as reported by the backend. RO stands for readout error.\label{tab:nairobi_prop}}
    }
\end{table}

\begin{table}[htbp!]
\centering{
\begin{tabular}{l r r | l r r}\hline\hline
Qubit & \multicolumn{2}{c|}{CNOT} & Qubit & \multicolumn{2}{c}{CNOT} \\
pair & Error & length & pair & Error & length \\ \hline
(0, 1) & 1.2 & 249 & (1, 2) & 0.5 & 391 \\
(4, 5) & 0.6 & 277 & (5, 6) & 0.8 & 306 \\
(1, 3) & 0.8 & 270 & (3, 5) & 1.2 & 235 \\ \hline\hline
\end{tabular}
\caption{Two-qubit gate error and length in \% and ${\rm ns}$, respectively, reported by \emph{ibm\_nairobi} on 22.12.2021.\label{tab:nairobi_gate}}
}
\end{table}

\begin{figure}[htbp!]
    \centering
    \includegraphics[width=0.9\columnwidth]{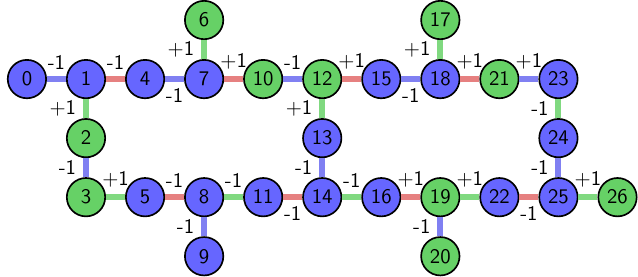}
    \caption{
    Coupling map of \emph{ibmq\_mumbai}.
    The edge coloring indicates which CNOT gates can be applied in parallel and the edge weights are those of the graph of the 27 node maximum cut problem.
    The node coloring shows the maximum cut, found with CPLEX, which has a value of 12.
    }
    \label{fig:mumbai_cmap}
\end{figure}

\clearpage

The 27 node graph is native to the coupling map of \emph{ibmq\_mumbai}, see Fig.~\ref{fig:mumbai_cmap}.
The results presented in the main text use CVaR aggregation with $\alpha=0.5$, i.e., the expectation value of $H_C$ is computed from the best 50\% of the shots.
Without CVaR aggregation we observe a noisy convergence at depth-one and no convergence at depth-two, see Fig.~\ref{fig:mumbai_v1}.
We allow SPSA 100 iterations at depth-one and 80 iterations at depth-two to avoid exceeding the maximum allowed time for a Qiskit Runtime program.
When SPSA terminates Qiskit returns the last values of $\boldsymbol{\beta}$ and $\boldsymbol{\gamma}$ and samples the state produced by the corresponding circuit.
For the 27 qubit problem instance with depth-two the optimizer does not converge properly. 
We therefore resample the circuit five times, each with $2^{14}$ shots, at the $\boldsymbol{\beta}$ and $\boldsymbol{\gamma}$ values that yielded the minimum energy. 
The average energy $-16.57\pm0.16$ of the five runs is shown as a star in Fig.~\ref{fig:mumbai}(a) and the best counts distribution is shown in Fig.~\ref{fig:mumbai}(b) of the main text.
The seven qubit problem on \emph{ibm\_nairobi} calibrates a readout-error mitigation matrix with $2^7$ circuits.
We therefore do not use readout error mitigation on the 27 qubit problem and leave it up to future work to investigate the effect of scalable readout error mitigation such as M3 on QAOA~\cite{Nation2021}.
The properties of the hardware and the gates as reported by \emph{ibmq\_mumbai} on the date the data were acquired are given in Tab.~\ref{tab:mumbai_prop} and~\ref{tab:mumbai_gate}.

\begin{figure}[htbp!]
\includegraphics[width=\columnwidth]{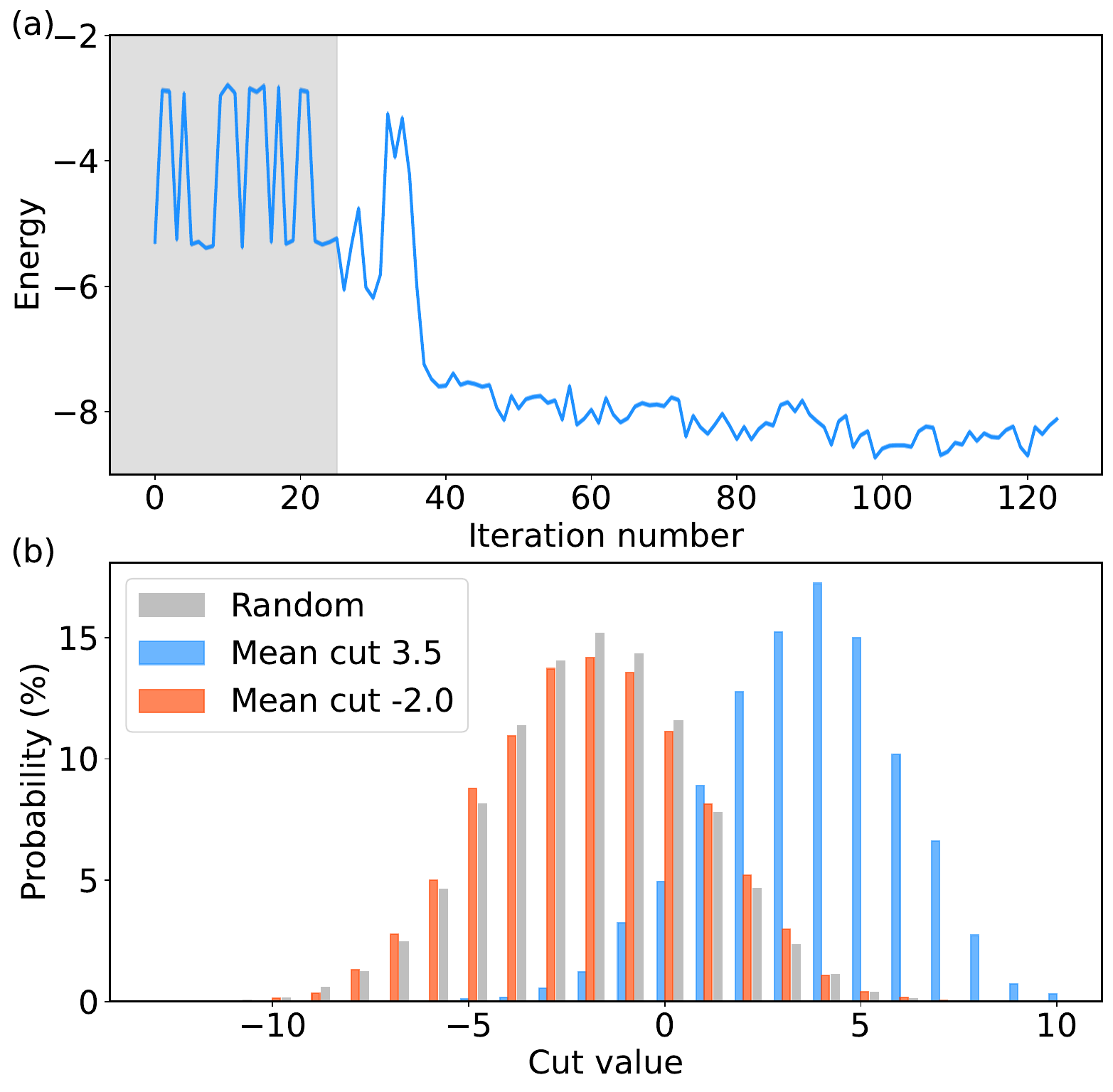}
\caption{
QAOA without CVaR on the graph native to the hardware of \emph{ibmq\_mumbai}.
At each energy evaluation 16384 shots are gathered.
(a) Energy as function of the iteration number for depth-one QAOA initialized at $\gamma=1$ and $\beta=0.5$.
In the first 25 iterations (gray area), SPSA calibrates the learning rate and perturbation based on the measured energy.
(b) Probability distribution of the best measured energy for depth-one (blue) compared to random sampling (gray) and depth-two (red).
\label{fig:mumbai_v1}
}
\end{figure}

\begin{table}[htbp!]
    \centering{
    \begin{tabular}{l r r r r} \hline\hline
 & $T_1$ & $T_2$ & $\sqrt{X}$ Error  & RO Error \\
Qubit & $\quad$($\mu{\rm s}$) & $\quad$($\mu{\rm s}$) & $\times 10^{-4}$ & \% \\ \hline
0 &  99 &  23 & 2.7 & 4.0 \\
1 & 168 & 171 & 1.9 & 2.1 \\
2 & 100 & 101 & 1.9 & 2.8 \\
3 & 119 & 192 & 2.4 & 1.9 \\
4 & 115 &  39 & 2.4 & 1.6 \\
5 & 121 & 137 & 2.9 & 1.8 \\
6 & 116 &  46 & 2.5 & 1.5 \\
7 & 170 &  71 & 2.4 & 2.5 \\
8 & 155 &  80 & 2.5 & 4.7 \\
9 &  69 &  98 & 6.0 & 2.4 \\
10 & 131 & 222 & 2.0 & 1.2 \\
11 &  96 &  88 & 1.8 & 3.0 \\
12 & 215 & 313 & 1.7 & 1.8 \\
13 & 165 &  40 & 1.6 & 3.7 \\
14 & 169 & 337 & 3.5 & 2.9 \\
15 &  65 &  84 & 2.3 & 2.1 \\
16 & 193 & 129 & 1.5 & 2.3 \\
17 & 135 & 178 & 2.2 & 1.5 \\
18 & 194 & 217 & 1.2 & 1.4 \\
19 & 211 & 272 & 1.7 & 1.3 \\
20 &  99 & 149 & 1.9 & 1.8 \\
21 &  65 &  81 & 3.5 & 2.9 \\
22 & 172 &  87 & 2.4 & 2.0 \\
23 & 109 &  30 & 4.8 & 11.9 \\
24 & 148 &  62 & 1.7 & 1.9 \\
25 & 180 &  71 & 1.6 & 1.5 \\
26 & 140 & 241 & 2.0 & 1.8 \\ \hline\hline
    \end{tabular}
    \caption{Hardware properties of \emph{ibmq\_mumbai} on 26.07.2022 as reported by the backend. RO stands for readout error.\label{tab:mumbai_prop}}
    }
\end{table}

\begin{table}[htbp!]
\centering{\small
\begin{tabular}{l r r | l r r}\hline\hline
Qubit & \multicolumn{2}{c|}{CNOT} & Qubit & \multicolumn{2}{c}{CNOT} \\
pair & Error & length & pair & Error & length \\ \hline
(0, 1) &   0.7 &  420 & (1, 2) & 1.5 & 704 \\
(1, 4) &   1.0 &  348 & (2, 3) & 0.7 & 391 \\
(3, 5) & -- &  356 & (4, 7) & 0.9 & 590 \\
(5, 8) &   2.3 & 1316 & (6, 7) & 0.7 & 249 \\
(7, 10) & 0.7 & 398 & (8, 9) & -- & 718 \\
(8, 11) & 4.0 & 590 & (10, 12) & 0.6 & 363 \\
(11, 14) & 1.9 & 370 & (12, 13) & 0.6 & 548 \\
(12, 15) & 0.6 & 370 & (13, 14) & 1.3 & 320 \\
(14, 16) & 1.4 & 526 & (15, 18) & 0.6 & 306 \\
(16, 19) & 1.9 & 512 & (17, 18) & 0.5 & 249 \\
(18, 21) & 0.7 & 476 & (19, 22) & 0.6 & 327 \\
(19, 20) & 0.5 & 370 & (21, 23) & 1.0 & 391 \\
(22, 25) & 0.6 & 448 & (23, 24) & 3.6 & 626 \\
(24, 25) & 0.5 & 434 & (25, 26) & 0.5 & 313 \\ \hline\hline
\end{tabular}
\caption{Two-qubit gate error and length in \% and ${\rm ns}$, respectively, reported by \emph{ibmq\_mumbai} on 26.07.2022.\label{tab:mumbai_gate}}
}
\end{table}

\clearpage

\end{document}